\documentclass[12pt,letterpaper]{article}

\usepackage[utf8]{inputenc}

\usepackage{amssymb,amsmath,amsthm,amsfonts}
\usepackage[pdftex,colorlinks=true,linkcolor=blue,citecolor=blue,urlcolor=red,unicode=true,hyperfootnotes=false,bookmarksnumbered]{hyperref}
\usepackage[margin=1truein]{geometry}

\usepackage{enumerate}

\theoremstyle{plain}
\newtheorem{theorem}{Theorem}[section]
\newtheorem{lemma}[theorem]{Lemma}
\newtheorem{claim}[theorem]{Claim}

\newtheorem{assumption}[theorem]{Assumption}
\theoremstyle{definition} 

\newtheorem{remark}[theorem]{Remark}
\newtheorem{definition}[theorem]{Definition}
\newtheorem*{notation}{Notation}

\makeatletter
\def\@gifnextchar#1#2#3{\let\@tempe#1\def\@tempa{#2}\def\@tempb{#3}%
  \futurelet\@tempc\@gifnch}

\def\@gifnch{\ifx\@tempc\@sptoken\let\@tempd\@tempb%
  \else\ifx\@tempc\@tempe\let\@tempd\@tempa\else\let\@tempd\@tempb\fi\fi\@tempd}

\def\SK@set#1{\left\{#1\right\}}
\def\SK@@set#1#2{\{#1\,:\,
    \begin{array}{@{}l@{}}#2\end{array}
\}}
\def\SK@mset#1{\left\{\!\!\left\{#1\right\}\!\!\right\}}
\def\SK@@mset#1#2{\{\!\!\{#1\,:\,
    \begin{array}{@{}l@{}}#2\end{array}
\}\!\!\}}
\def\BIG@set#1{\Big\{#1\Big\}}
\def\BIG@@set#1#2{\Big\{#1\:\Big|\:
    \begin{array}{@{}l@{}}#2\end{array}
\Big\}}
\newcommand{\Set}[1]{\@gifnextchar\bgroup{\SK@@set{#1}}{\SK@set{#1}}}
\newcommand{\Mset}[1]{\@gifnextchar\bgroup{\SK@@mset{#1}}{\SK@mset{#1}}}
\newcommand{\Bigset}[1]{\@gifnextchar\bgroup{\BIG@@set{#1}}{\BIG@set{#1}}}
\makeatother

\newcommand{\refeq}[1]{(\ref{eq:#1})}
\newcommand{\of}[1]{\left( #1 \right)}
\newcommand{\ofsq}[1]{\left[ #1 \right]}

\newcommand{\bG}{\mathbf{G}}

\newcommand{\rg}{\mathbf{G}_n}

\newcommand{\prob}[1]{\mathsf{P}[ #1 ]}
\newcommand{\probb}[1]{\mathsf{P}\left[ #1 \right]}
\newcommand{\cprob}[2]{\mathsf{P}[ #1 \mid #2 ]}
\newcommand{\cprobb}[2]{\mathsf{P}\left[ #1 \;\middle\vert\; #2 \right]}
\newcommand{\Mod}[1]{\ (\mathrm{mod}\ #1)}
\newcommand{\uu}[1]{\mathbf{u}_{#1}}
\newcommand{\cQ}{\mathcal{Q}}
\newcommand{\cG}{\mathcal{G}}
\newcommand{\cA}{\mathcal{A}}
\newcommand{\seq}[1]{\mathsf{#1}}
\newcommand{\lds}{\mathit{lds}}
\newcommand{\bv}{\mathrm{bv}}

\title{On anti-stochastic properties of unlabeled graphs}

\author{
Sergei Kiselev\thanks{kiselev.sg@gmail.com}
\and
Andrey Kupavskii\thanks{CNRS, France; kupavskii@ya.ru}
\and
Oleg Verbitsky\thanks{Institut f\"ur Informatik,
Humboldt-Universit\"at zu Berlin, Unter den Linden 6, D-10099 Berlin; verbitsk@informatik.hu-berlin.de
}
\and
Maksim Zhukovskii \thanks{The University of Sheffield, UK; m.zhukovskii@sheffield.ac.uk}
}

\date{}

\begin{document}

\maketitle

\begin{abstract}
We study vulnerability of a uniformly distributed random graph to an
attack by an adversary who aims for a global change of the distribution while
being able to make only a local change in the graph.
We call a graph property~$A$ \emph{anti-stochastic} if
the probability that a random graph $G$ satisfies $A$ is small but, with high probability, there is a small
perturbation transforming $G$ into a graph satisfying $A$.
While for labeled graphs such properties are easy to obtain from binary covering codes,
the existence of anti-stochastic properties for unlabeled graphs is not so evident.
If an admissible perturbation is either the addition or the deletion of one edge,
we exhibit an anti-stochastic property that is satisfied by a random unlabeled graph of order $n$
with probability $(2+o(1))/n^2$, which is as small as possible.
We also express another anti-stochastic property in terms of the degree sequence of a graph.
This property has probability $(2+o(1))/(n\ln n)$, which is optimal up to factor of~2.
\end{abstract}

\section{Introduction}

The asymptotic properties of a random graph are the subject of a rich and comprehensive theory \cite{Bollobas-b,Janson}.
Specifically, let $\rg$ be a graph chosen equiprobably from among all graphs on the vertex set $\{1,\ldots,n\}$.
Identifying a graph property with the set of all graphs possessing this property,
we say that $\rg$ has a property $P$ \emph{with high probability (whp)} or \emph{asymptotically almost surely} if
$\prob{\rg\in P}=1-o(1)$ as $n$ increases. 

If the ``error probability'' $o(1)$ is very small, then the
definition above is stable with respect to local perturbations of $\rg$. To be specific, 
here and below a perturbation means adding one edge to a graph or deleting one edge from it.
If $\prob{\rg\in P}=1-o(1/n^2)$, then the union bound implies that, whp, $P$ holds not only for $\rg$
but even for each of the ${n\choose2}$ perturbed versions of $\rg$. In other words, the property $P$
is robust with respect to the following adversarial attack. An adversary receives a random graph $\rg$
and is allowed to change the (non)adjacency of a single pair of vertices in $\rg$. Whatever he does,
the modified graph $\rg'$ still satisfies $P$ whp.

While this was a resilience example, our purpose in this paper is to study vulnerability issues
for this kind of attack. We address a scenario when an adversary is able to modify $\rg$
in a single vertex pair so
that, whp, the corrupted graph $\rg'$ has a property which is unlikely for a random graph. More precisely,
we call a property $A$ \emph{anti-stochastic} if the following two conditions are true:
\begin{itemize}
\item 
$\prob{\rg\in A}=o(1)$, and
\item 
there is an adversary such that $\prob{\rg'\in A}=1-o(1)$.
\end{itemize}
More formally, by an adversary we understand an arbitrary function $f$ which if applied
to a graph $G$, produces a graph $f(G)$ such that $G$ and $f(G)$ differ by at most a single edge.
Thus, $\rg'=f(\rg)$ is a random variable which we observe instead of a uniformly distributed~$\rg$.

Our interest in anti-stochastic properties is motivated by the fact that they yield a
conceptual formalization of the \emph{global} damage effect on a source of random graphs that can be caused
by a malicious adversary allowed to make only a \emph{local} perturbation in a graph he accesses.
Adversarial attacks on a random graph are studied in \cite{BollobasR03,BollobasR03a,FlaxmanFV07} 
focusing on the question on how many vertices can or must be deleted
in order to make a dynamically evolving random graph highly disconnected.
In general, resilience and vulnerability of graphs have been studied in network science
in many various contexts; we refer to the recent survey \cite{Schaeffer21} for overview of the
concepts and results in this large research area.

Let $A_n$ denote the set of $n$-vertex graphs with property $A$, and set $N={n\choose2}$.
Since a graph in $A_n$ can be a perturbed version of at most $N$ graphs,
the second condition in the above definition implies that 
$
|A_n|(N+1)\ge(1-o(1))2^N
$
for any anti-stochastic property.
It immediately follows that if $A$ is anti-stochastic, then 
\begin{equation}
  \label{eq:lower}
\prob{\rg\in A}\ge(2-o(1))/n^2.  
\end{equation}

This argument readily reveals a notable source of anti-stochastic properties.
A set $C_N\subset\{0,1\}^N$ is a \emph{covering code} \cite{CohenHLL-b}
if every string in $\{0,1\}^N$ is within Hamming distance 1 of some string in $C_N$.
For $N={n\choose2}$, we can identify the graphs on the vertex set $\{1,\ldots,n\}$
with the binary strings of length $N$. Thus,
if the density $|C_N|/2^N$ tends to 0 as $N$ increases, then the binary code
can be seen as an anti-stochastic property. The lower bound \refeq{lower}
turns out to be tight as there are covering codes of asymptotically optimal density $(1+o(1))/N$; 
see \cite{KabatyanskiiP88} or \cite[Ch.~12.4]{CohenHLL-b}.

Being a natural combinatorial concept in the realm of strings,
covering codes can hardly be considered natural graph properties.
As a minimum criterion for a graph property to be natural, we require
that it should be isomorphism invariant, that is, it should hold or not hold
for every two isomorphic graphs simultaneously. Our first result meets this expectation, 
implying that the damage caused by a combinatorially optimal adversary can be, in a way, conscious.

\begin{theorem}\label{thm:unlab}
There is an isomorphism invariant anti-stochastic property holding for a random graph of order $n$ with probability $(2+o(1))/n^2$.
\end{theorem}

Our proof of Theorem \ref{thm:unlab} uses the existence of covering codes with
optimal density. Ensuring the invariance under graph isomorphism is, however, a subtle business.
If we identify a binary code with the corresponding set of labeled graphs
and just take the closure of this set under isomorphism, then we cannot exclude that
this closure will have too high density, violating the first condition of an anti-stochastic
property. To rectify this problem, we use the following strategy.
\begin{itemize}
\item
For a graph $G$, we define a set $W$ of vertices of $G$ in terms of their degrees.
\item
We define another set of vertices $W'$ such that $W$ and $W'$ are disjoint.
For a vertex $w\notin W$, its membership in $W'$ is determined by the degree of $w$
in the subgraph of $G$ induced on the complement of $W$. Moreover, $W'$ is split
into ten parts $W'_1,\ldots,W'_{10}$ according to the vertex degrees in this induced subgraph.
\item
If $G$ is chosen randomly, then every two vertices in $W$ have, whp, differently many
neighbors in $W'_i$ for some $i\le10$, which determines a canonical labeling of $W$.
\item
Somewhat loosely speaking, the subgraphs of $G$ induced on $W$ and $W'$ are almost
uniformly distributed and independent from each other as well as from the adjacency
pattern between $W$ and $W'$. This ensures that, whp, the subgraph induced on $W$
remains uniformly distributed also with respect to the aforementioned canonical labeling.
In this way we can simulate the property of a binary word to belong to a covering code
by an isomorphism invariant graph propery.
\end{itemize}
The above strategy resembles the classical canonization procedure
for almost all graphs due to Babai, Erd\H{o}s, and Selkow \cite{BabaiES80}.
Note that we cannot use canonization of a random graph directly because
we need the graph to be uniformly distributed after relabeling, which makes our problem
more sofisticated.

Our second construction of an isomorphism invariant anti-stochastic property is, in a certain sense,
more natural as it is defined solely
in terms of the degree sequence of a graph. Note that any condition on the degree sequence
defines an isomorphism invariant graph property.

\begin{theorem}\label{thm:degrees}\hfil
  \begin{enumerate}[\bf 1.]
  \item 
There is an anti-stochastic property expressible in terms of the degree sequence
that holds for a random graph of order $n$ with probability $(2+o(1))/(n\ln n)$.  
  \item 
On the other hand, every such anti-stochastic property has probability at least $(1-o(1))/(n\ln n)$.   
  \end{enumerate}
\end{theorem}

The proof uses the approximability of the vertex degrees in a random graph by independent binomial variables~\cite{McKayW97}. 
Moreover, the lower bound is based on the fact that the number of graphs with a given degree sequence does not change much if 
the ``frequent'' degrees are slightly changed. We derive this from the asymptotics of the number of graphs with a given 
degree sequence~\cite{LiebenauWormald}. 
The proof of the upper bound, in a sense, simulates the well-known direct sum construction for covering codes; see, e.g., 
\cite[Chapter 3.2]{CohenHLL-b}. 
We split the set of ``frequent'' vertex degrees in a graph into two parts and associate with them two 
0-1-vectors $Y^\downarrow$ and $Y^\uparrow$ of sizes $\delta_1$ and $\delta_2$ respectively   
such that any change of the adjacency relation between a vertex from the first part and a vertex from the second part affects
a single coordinate in each of the vectors $Y^\downarrow$ and $Y^\uparrow$. 
Then we fix two optimal covering codes $\mathcal{S}^{\downarrow}$ and $\mathcal{S}^{\uparrow}$ of lengthes $\delta_1$ and $\delta_2$ respectively
and consider the graph property saying that $Y^\downarrow\in\mathcal{S}^{\downarrow}$ and $Y^\uparrow\in\mathcal{S}^{\uparrow}$.
This property is anti-stochastic and is expressed solely in terms of the vertex degrees.
It has probability $(4+o(1))/(n\ln n)$ because, as we show, $Y^\downarrow$ and $Y^\uparrow$ can be handled
like almost uniformly distributed random vectors
and because counting the ``frequent'' degrees yields $\delta_1,\delta_2=\frac{1}{2}\sqrt{n\ln n}(1+o(1))$.
We are able to decrease this probability by a factor of 2 by bringing into play the possibility to choose 
between insertion or deletion of an edge.

We conclude discussion of our results with two complexity-theoretic comments.

\begin{remark}
It is natural to ask which computational power must the adversary have
in order to corrupt the random graph by enforcing an anti-stochastic property~$A$.
Obviously, it would be enough for him to be able to recognize whether or not
a given graph satisfies $A$. 
The decision complexity of the property  constructed in the proof
of Theorem \ref{thm:unlab} is no more than the decision complexity of the
covering code used in the construction. If $n=2^k-1$, we can use the Hamming code,
which is perfect and can serve as a covering code. The membership
in the Hamming code is efficiently recognizable. If $n$ is not of this kind,
we can expand the closest Hamming code to the length $n$ just by appending the missing bits
in all possible ways, which retains
the polynomial-time complexity. However, the probability of the corresponding
anti-stochastic property becomes twice higher, i.e., $(4+o(1))/n^2$ instead
of $(2+o(1))/n^2$. If the adversary is not content with this, he has to
use an asymptotically optimal covering code.
Such a code is suggested by Kabatyanskii and Panchenko \cite{KabatyanskiiP88}.
Since their construction uses randomization, the recognition complexity
of this code seems to be a subtle issue.
\end{remark}

\begin{remark}\label{rem:descr}
A straightforward inspection of the proof of Theorem \ref{thm:unlab} reveals
that the anti-stochastic property $A$ we constructed is not just isomorphism invariant
but it can even be expressed in terms of 4-iterated degree sequences; see Remark \ref{rem:it-degrees} for technical details.
This implies that $A$ has low descriptive complexity in the sense that it is expressible in the 2-variable infinitary logic $C^2_{\infty\omega}$, even with bounded quantifier depth. We refer the reader to \cite{ImmermanL90} for this logical interpretation. 

In other words, if graphs $G$ and $G'$ are indistinguishable by the classical
degree refinement procedure (also called color refinement or the 1-dimensional Weisfeiler-Leman algorithm)
in 4 rounds and $G\in A$, then $G'\in A$ as well.
Recall that the average-case canonization of Babai, Erd\H{o}s, and Selkow \cite{BabaiES80}
is nothing else as the 2-round degree refinement. In \cite{BabaiES80} it is shown
that a random graph is canonizable in this way with probability $1-n^{-c}$ for a positive constant $c$.
This probability becomes exponentially close to 1 for the 3-round degree refinement \cite{BabaiK79}.
On the other hand, a random graph is canonized by its degree sequence with very small
probability. It is curious that our Theorems \ref{thm:unlab} and \ref{thm:degrees}
show a quantitative difference between different numbers of degree iterations
in a quite different setting. 
\end{remark}

\begin{remark}
Finally, note that so far we have considered \emph{labeled} graphs (whose vertices are labeled by $1,\ldots,n$).
An \emph{unlabeled} graph can be defined formally as an isomorphism class of labeled graphs.
Theorem \ref{thm:unlab} can be translated into virtually the same statement saying
that there exists an anti-stochastic property of unlabeled
graphs holding with probability $(2+o(1))/n^2$. Here, a random graph is chosen equiprobably
from among all unlabeled graphs of order $n$. An isomorphism invariant property $P$ of labeled graphs
gives rise to a property of unlabeled graphs, which we denote by $\hat P$.
The unlabeled version of Theorem \ref{thm:unlab} follows from the well-known fact \cite[Ch.~9]{HararyP-b}
that, for random graphs of order $n$, 
$$
(1-o(1))\prob P-O\of{\frac{n^2}{2^n}}\le\prob{\hat P}\le(1+o(1))\prob P+O\of{\frac{n^2}{2^n}}.
$$
The former of these two inequalities also shows that the lower bound \refeq{lower} holds true
as well for anti-stochastic properties of unlabeled graphs.
\end{remark}

The results of this paper were presented in preliminary form in the 48th International Workshop on 
Graph-Theoretic Concepts in Computer Science~\cite{WG22}.

\section{Preliminaries}

The vertex set of a graph $G$ is denoted by $V(G)$.
If $U\subset V(G)$, then we write $G|_U$ to denote the subgraph of $G$ induced on the set of vertices $U$.
The degree of a vertex $v\in V(G)$ is denoted by~$\deg_G(v)$.

 For the uniformly distributed random graph $\rg$
on $n$ vertices, it is supposed that $V(\rg)=[n]$, where $[n]=\{1,\ldots,n\}$.
In some contexts, a set $A$ of graphs on $[n]$ can be identified with the event $\rg\in A$,
whose probability $\prob{\rg\in A}$ will then be denoted by~$\prob A$.

Let us recall several probabilistic facts that we use in our proofs. Let $\xi$ be a Binomial random variable with parameters $n$ and $p$,
that is, $\xi=\sum_{i=1}^n\xi_i$ where $\xi_i$'s are mutually independent and, for each $i$, we have $\xi_i=1$ with probability 1 and
$\xi_i=0$ with probability $1-p$. We use the notation $\xi\sim\mathrm{Bin}(n,p)$ when $\xi$ has this distribution. 
The following estimates, which are special cases of Bernstein's inequalities,
show that $\xi$ is well-concentrated around its expectation.

\begin{lemma}[Chernoff's bounds; see, e.g.,~{\cite[Theorem 2.1]{Janson}}]
If $\xi\sim\mathrm{Bin}(n,p)$, then
$$
 \prob{\xi\leq np-t}\leq\exp\left(-\frac{t^2}{2np}\right)\text{ and }
 \prob{\xi\geq np+t}\leq\exp\left(-\frac{t^2}{2np+t/3}\right)
$$
for every $t\geq 0$.
\label{Chernoff}
\end{lemma}

We remark that if $\xi\sim\mathrm{Bin}(n,1/2)$, then
\begin{equation}
\prob{\xi=k}\leq\prob{\xi=\lfloor n/2\rfloor}=\sqrt{\frac{2}{\pi n}}(1+o(1))
\label{binomial_bound}
\end{equation}
for each $k$. The latter equality follows from the de Moivre--Laplace limit theorem.

Recall that the \emph{characteristic function} of a random variable $X$ is defined as $\phi(t)={\sf E}e^{itX}$ where $t\in\mathbb{R}$.
For an $n$-dimensional random vector $X=(X_1,\ldots,X_n)$, this generalizes to $\phi(t_1,\ldots,t_n)={\sf E}e^{i\sum_{k=1}^nt_kX_k}$.
Recall that $\phi(t)=(1-p+pe^{it})^n$ for $X\sim\mathrm{Bin}(n,p)$.
The first part of the following lemma can be found in~\cite[Theorem~1]{Herschkorn}. The second part is a simple
extension to the multi-dimensional case.

\begin{lemma}\label{lem:character}
Let $m$ be a positive integer.
  \begin{enumerate}[\bf 1.]
  \item
(Herschkorn~\cite{Herschkorn})
Let $X$ be an integer-valued random variable with characteristic function~$\phi$. Then
$$
\prob{m\text{ divides }X}=\frac{1}{m}\sum_{j=0}^{m-1}\phi\left(\frac{2\pi j}{m}\right).
$$  
\item
Let $X=(X_1,\ldots,X_n)$ be an integer-valued random vector with characteristic function $\phi$, then
$$
\prob{m\text{ divides }X_k\text{ for every }k=1,\ldots,n}=
\frac{1}{m^n}\:\sum_{j_1=0}^{m-1}\ldots\sum_{j_n=0}^{m-1}\phi\left(\frac{2\pi j_1}{m},\ldots,\frac{2\pi j_n}{m}\right).
$$
\end{enumerate}
\end{lemma}

We first apply part 1 of Lemma \ref{lem:character} to show that the remainders of a random variable $\mathrm{Bin}(n,1/2+o(1))$
modulo $m$ are evenly distributed. After this, part 2 of Lemma \ref{lem:character} will be used to prove
that if take the remainders modulo $m$ of the vertex degrees in a random graph, then
the resulting sequences are evenly distributed (Lemma \ref{lm:graphs_division} below).

\begin{lemma}
Let $\xi\sim\mathrm{Bin}(n,1/2+o(1))$. For every positive integer $m$ and $r\in\{0,\ldots,m-1\}$,
$$
\left|\prob{\xi\equiv r\Mod m}-\frac{1}{m}\right|=O(e^{-cn}),
$$
where the constant $c>0$ as well as the constant absorbed in the big O notation depend only on~$m$.
\label{lem:degrees_divisible}
\end{lemma}

\begin{proof}
Consider the random variable $\xi-r$ and its characteristic function 
$$
\phi(t)={\sf E}e^{it(\xi-r)}=e^{-itr}{\sf E}e^{it\xi}=e^{-itr}\of{\frac{1+e^{it}+o(1)}{2+o(1)}}^n.
$$
By part 1 of Lemma~\ref{lem:character}, 
\begin{multline*}
\left|\prob{\xi\equiv r\Mod m}-\frac{1}{m}\right|=\left|\frac{1}{m}\sum_{j=1}^{m-1}\phi\left(\frac{2\pi j}{m}\right)\right|\\
=
 \left|\frac{1}{m}\sum_{j=1}^{m-1} e^{-2ijr\pi/m}
\left(
\frac{1+\cos(2\pi j/m)+i\sin(2\pi j/m)+o(1)}
{2+o(1)}
\right)^n\right|\\
\le
 \frac{1}{m}\sum_{j=1}^{m-1}
\left|
\frac{1+\cos(2\pi j/m)+i\sin(2\pi j/m)+o(1)}
{2+o(1)}
\right|^n.
\end{multline*}
Note that
$$
1+\cos\frac{2\pi j}{m}+i\sin \frac{2\pi j}{m}=2\cos^2\frac{\pi j}{m}+2i\sin \frac{\pi j}{m}\cos\frac{\pi j}{m}=
2\cos\frac{\pi j}{m}e^{ij\pi/m}.
$$
It follows that if $1\le j<m$, then the term
$$
\left|1+\cos\frac{2\pi j}{m}+i\sin \frac{2\pi j}{m}+o(1)\right|=
2\left|\cos\frac{\pi j}{m}\right|+o(1) \le
2\cos\frac{\pi}{m}+o(1)
$$
is bounded away from 2, which completes the proof.  
\end{proof}

\begin{lemma}
Fix an odd integer $m\ge3$ and let $r_v\in\{0,1,\ldots,m-1\}$ for $v=1,\ldots,n$. Then
$$
\left|\prob{\deg_{\rg}(v)\equiv r_v\Mod m\text{ for every }v=1,\ldots,n}-\frac1{m^n}\right|=O(e^{-cn}),
$$
where the constant $c>0$ as well as the constant absorbed in the big O notation depend solely on~$m$.
\label{lm:graphs_division}
\end{lemma}

\begin{proof}
Let $Y=(Y_1,\ldots,Y_n)$ be the vector of vertex degrees of $\rg$. 
Let $R=(r_1,\ldots,r_n)$ and $X=Y-R$. Denote the characteristic functions of $X$ and $Y$ by $\phi_X$ 
and $\phi_Y$ respectively and note that $\phi_X(t_1,\ldots,t_n)=e^{-i\sum_{v=1}^nr_vt_v}\phi_Y(t_1,\ldots,t_n)$. 
By part 2 of Lemma \ref{lem:character}, the probability that $Y_v\equiv r_v\Mod m$ for all $v$ is equal to
$$
\frac{1}{m^n}\:\sum_{x\in\{0,1,\ldots,m-1\}^n}\phi_X\left(\frac{2\pi x}{m}\right)=
\frac{1}{m^n}\:\of{1+\sum_{x\ne0}\phi_X\left(\frac{2\pi x}{m}\right)}
$$
where $2\pi x/m$ is an $n$-dimensional real vector and the equality is just due to the observation
that $\phi_X(0)=1$ for the $n$-dimensional zero vector. Noting that
$$
\left|\sum_{x\ne0}\phi_X(2\pi x/m)\right|\le\sum_{x\ne0}\left|\phi_X(2\pi x/m)\right|\le\sum_{x\ne0}\left|\phi_Y(2\pi x/m)\right|,
$$ 
we have to prove that the last sum is exponentially small. 
To this end, we bound the term $|\phi_Y(2\pi x/m)|$
from above for each non-zero vector $x\in\{0,1,\ldots,m-1\}^n$.

For distinct $u,v\in[n]$, let $\xi_{u,v}$ denote the indicator random variable of the presence of the edge $\{u,v\}$ in $\rg$. 
For a real vector $t=(t_1,\ldots,t_n)$, we have
$$
 \phi_Y(t)={\sf E}e^{i\sum_{u=1}^n t_uY_u}={\sf E}e^{i\sum_{u=1}^n\sum_{v\neq u}t_u\xi_{u,v}}=
 {\sf E}e^{i\sum_{u<v}(t_u+t_v)\xi_{u,v}}=\prod_{u<v}\left(\frac{1}{2}+\frac{1}{2}e^{i(t_u+t_v)}\right).
$$
Set $\alpha=\max_{1\le j<m}|1+e^{2\pi ij/m}|$ and note that $\alpha<2$.
Let $n_0(x)$ be the number of zeros in $x$.
Note that $|\frac{1}{2}+\frac{1}{2}e^{i(t_u+t_v)}|\le1$. Moreover, this number does not 
exceed $\alpha/2$ for $t=2\pi x/m$ such that exactly one of the coordinates $x_u$ and $x_v$ is equal to $0$.
It follows that
$$
|\phi_Y(2\pi x/m)|\le(\alpha/2)^{n_0(x)(n-n_0(x))}.
$$
Now, fix $k$ to be the smallest integer such that $(\alpha/2)^k<1/m$. 
Since $x\ne0$, we conclude that
$$
|\phi_Y(2\pi x/m)|\le
\begin{cases}
(\alpha/2)^{n-1}&\text{if }n_0(x)\geq n-k\\
(\alpha/2)^{k(n-k)}&\text{if }n/2\leq n_0(x)<n-k.
\end{cases}
$$
Consider the case that $n_0(x)<n/2$. Since each of the more than $n/2$ non-zero coordinates of $x$
can take on at most $m-1$ values, there are at least $(m-1){\lceil n/(2(m-1)) \rceil \choose 2}$
pairs $(u,v)$ with $u<v$ and $x_u=x_v\neq 0$. For such a pair, the assumption that $m$ is odd implies
that the sum $x_u+x_v$ is not divisible by $m$. It follows that 
$|\phi_Y(2\pi x/m)|\leq(\alpha/2)^{(m-1){\lceil n/(2(m-1))\rceil \choose 2}}$.

Taking these estimates into account, we get 
\begin{multline*}
\sum_{x\ne0}\left|\phi_Y\left(\frac{2\pi x}{m}\right)\right| < k {n\choose k}m^k\left(\frac{\alpha}{2}\right)^{n-1} \\
+m^n\left(\left(\frac{\alpha}{2}\right)^{k(n-k)}+\left(\frac{\alpha}{2}\right)^{(m-1){\lceil n/(2(m-1))\rceil \choose 2}}\right)
=O\of{\of{m(\alpha/2)^k}^n},
\end{multline*}
where the sum goes over all non-zero $x\in\{0,1,\ldots,m-1\}^n$. This completes the proof.
\end{proof}

We will also need the result of McKay and Wormald~\cite{McKayW97} stating that the degrees of $\rg$ are almost independent. 

\begin{theorem}[McKay and Wormald~\cite{McKayW97}]
Let $\mathbf{p}$ be a normal random variable with mean $1/2$ and variance $1/(4n(n-1))$, truncated to $(0,1)$. 
For each $x\in(0,1)$, let $\eta_1(x),\ldots,\eta_n(x)$ be independent binomial random variables
such that $\eta_i(x)\sim\mathrm{Bin}(n-1,x)$ for each $i=1,\ldots,n$. Moreover, assume that
the random vector $(\eta_1(x),\ldots,\eta_n(x))$ is independent of $\mathbf{p}$.
Let $\xi_1,\ldots,\xi_n$ be the degrees of the vertices $1,\ldots,n$  in $\rg$.
Fix a real $\alpha>0$ and a set $S_n\subseteq\{0,1,\ldots,n-1\}^n$ for each $n$. Then
$$
 \prob{(\xi_1,\ldots,\xi_n)\in S_n}=(1+o(1))\;
 \cprobb{(\eta_1(\mathbf{p}),\ldots,\eta_n(\mathbf{p}))\in S_n }{ \sum_{i=1}^n \eta_i(\mathbf{p})\text{ is even}}
+o\left(n^{-\alpha}\right).
$$
\label{th:MW}
\end{theorem}

\section{Proof of Theorem \ref{thm:unlab}}

Let $k>11$ be an odd integer non-divisible by 11. For a graph $G$, we define $U(G)$ to be 
the set of vertices of $G$ whose degrees are divisible by $k$.
For $r=0,1,\ldots,10$, let $U_r(G)$ denote the set of those vertices in $U(G)$ 
whose degrees in $G|_{U(G)}$ are congruent to $r$ modulo $11$.
We also set $W(G)=V(G)\setminus U(G)$ and $R(G)=U(G)\setminus U_0(G)$.
In what follows, an important role will be played by the partition
\begin{equation}
  \label{eq:partition}
V(G)=W(G)\cup U_0(G)\cup R(G).  
\end{equation}
For notational simplicity, we suppress the dependence of this partition on $k$.
The value of the parameter $k$ is supposed to be fixed until the final step of the proof.

For the random graph $\rg$, the partition \refeq{partition} translates in
$$
[n]=\mathbf{W}\cup\mathbf{U}_0\cup\mathbf{R},
$$
where $\mathbf{W}=W(\rg)$, $\mathbf{U}_r=U_r(\rg)$, and $\mathbf{R}=R(\rg)$. Also, $\mathbf{U}=U(\rg)$. 

\subsection{The property}\label{ss:prop}

For an $n$-vertex graph $G$, we will suppose that $V(G)=[n]$.
Let $Q$ be an anti-stochastic property of labeled graphs.
More specifically, for each $n$ we fix a covering code in $\{0,1\}^{{n\choose 2}}$
of asymptotically optimal density. An $n$-vertex graph $G$
belongs to $Q$ if the ${n\choose 2}$-dimensional vector 
of adjacencies of $G$ belongs to the code. By~\cite{KabatyanskiiP88}, we may assume that 
$\prob{\rg\in Q}\leq(2+1/k)/n^2$ for large enough~$n$. 

We say that $R(G)=U_1(G)\cup\ldots\cup U_{10}(G)$ \emph{resolves} $W(G)$ if every two distinct vertices in $W(G)$
have differently many neighbors in $U_r(G)$ for some $r\in[10]$. More specifically, for $v\in W(G)$, let
$\vec{d}(v)=(d_1(v),\ldots,d_{10}(v))$ where $d_r(v)$ denotes the number of neighbors of $v$ in $U_r(G)$.
Then $R(G)$ resolves $W(G)$ if $\vec{d}(v)\ne \vec{d}(u)$ for any distinct $u,v\in W(G)$.
If this is the case, consider the lexicographical order on $\{0,1,\ldots,n-1\}^{10}$ and relabel
the vertices in $W(G)$ by the integers $1,\ldots,|W(G)|$ according to this order.
This results in an isomorphic copy of $G|_{W(G)}$, and we will say that the subgraph $G|_{W(G)}$
is \emph{canonically relabeled}.

Let $B$ denote the set of graphs $G$ such that $R(G)$ resolves $W(G)$.
We define the property $Q_k$ by setting $G\in Q_k$ if
\begin{itemize}
\item either $G\notin B$,
\item or $G\in B$ and the canonically relabeled subgraph $G|_{W(U)}$ belongs to~$Q$.
\end{itemize}
Note that $Q_k$ is isomorphism invariant. Indeed, this is obvious for the property $B$.
Now, if $G$ satisfies the second condition above and $G'\cong G$, then $G'\in B$ too.
Let $f$ be an isomorphism from $G$ to $G'$. Note that $f$ induces an isomorphism from $G|_{W(G)}$
to $G'|_{W(G')}$ preserving the canonical labels. As a consequence, $G'$ also satisfies the second 
condition in the definition of~$Q_k$.

The proof of Theorem \ref{thm:unlab} consists of three parts.
\begin{enumerate}[1.]
\item 
We will prove that $Q_k$ has small probability, specifically,
$\prob{\rg\in Q_k}\le\of{1+\frac3k}\,\frac2{n^2}$ for large enough~$n$.
\item 
Then, we will prove that $Q_k$ is close to an anti-stochastic property in the sense that
an adversary is able to transform $\rg$ in $\rg'$ such that $\prob{\rg'\in Q_k}>1-\frac4k$ for large enough~$n$.
Recall that, given a graph $G$, the adversary is allowed to change the adjacency of at most one vertex pair in~$G$.
\item 
These facts will allow us to combine a sequence of properties $Q_k$ into a single anti-stochastic property~$Q^*$.
\end{enumerate}

Technically, part 1 of this plan will be accomplished by showing that
\begin{itemize}
\item 
$B$ holds with probability $1-o(1/n^2)$, and
\item 
the canonically relabeled subgraph $\rg|_{\mathbf{W}}$ remains almost uniformly distributed.
\end{itemize}

We now introduce notation that will be used in the next subsections.
\begin{notation}
For a set $X$, we write $\mathbf{G}_X$ to denote the uniformly distributed random
graph on the vertex set $X$. Moreover, if $X\cap Y=\emptyset$, then $\mathbf{G}_{X\times Y}$ stands for
the uniformly distributed random bipartite graph with vertex classes $X$ and $Y$.
\end{notation}

Our main technical tool will be a lemma about asymptotical independence and uniformity of the subgraph of $\rg$ 
induced on $\mathbf{R}$, the subgraph induced on $\mathbf{W}$, 
and the induced bipartite graph between $\mathbf{R}$ and $\mathbf{W}$.
Note that all this is equivalent to asymptotical uniformity of the subgraph~$\rg|_{\mathbf{W}\cup\mathbf{R}}$.
Somewhat loosely speaking, our nearest goal is to show that the random graphs $\rg|_{\mathbf{W}\cup\mathbf{R}}$
and $\mathbf{G}_{W\cup R}$ have almost the same distribution under the condition that $\mathbf{W}=W$
and $\mathbf{R}=R$.

\subsection{Distribution of induced subgraphs}

\begin{assumption}\label{assump}
We fix a real $\varepsilon\in(0,1)$ and suppose that $U_0,U_1,\ldots,U_{10}$ are disjoint subsets of $[n]$ such that
$$
\frac{n}{11k}(1-\varepsilon)<|U_i|<\frac{n}{11k}(1+\varepsilon).
$$
We also set $U=\bigcup_{i=0}^{10} U_i$, $W=[n]\setminus U$, and $R=U\setminus U_0$.
Moreover, ${\vec U}=(U_0,\ldots,U_{10})$. 
\end{assumption}
Note that the sets in Assumption \ref{assump}, in contrast to the sets $U_r(G)$, $W(G)$ etc.\ defined above,
are considered irrespectively of any graph~$G$.

Recall that $\mathbf{U}_r=U_r(\rg)$ and set $\vec{\mathbf{U}}=(\mathbf{U}_0,\ldots,\mathbf{U}_{10})$. 
We show that, conditioned on $\mathbf{U}=U$ and on $\vec{\mathbf{U}}={\vec U}$ respectively, 
the graphs $\rg|_U$ and $\rg|_{W\cup R}$ are `almost' uniformly distributed.

\begin{lemma}
Under Assumption \ref{assump} we have the following equalities.
\begin{enumerate}[\bf 1.]
\item For every property $A$ of graphs on $U$,
$$
 \cprobb{\rg|_{U}\in A}{\mathbf{U}=U}=(1+o(1))\,\probb{\mathbf{G}_U\in A}.
$$
\item For every property $A$ of graphs on $[n]\setminus U_0$,
$$
 \cprobb{\rg|_{W\cup R}\in A}{\vec{\mathbf{U}}={\vec U}} = 
 (1+o(1))\,\probb{\mathbf{G}_{W\cup R}\in A}.
$$
\end{enumerate}
\label{lm:asympt_uniform}
\end{lemma}

\begin{proof}
For graphs $G_1$ and $G_2$ we define their union $G_1\cup G_2$ as the graph with
$V(G_1\cup G_2)=V(G_1)\cup V(G_2)$ and $E(G_1\cup G_2)=E(G_1)\cup E(G_2)$,
where $E(G)$ denotes the edge set of a graph~$G$.

Let $u=|U|$, $u_0=|U_0|$, and $w=|W|$.

\medskip

\noindent
\textbf{1.}  
Fix an arbitrary graph $H$ with $V(H)=U$. It suffices to prove that
\begin{equation}\label{eq:HU}
 \cprobb{\rg|_{U}=H}{\mathbf{U}=U}=(1+o(1))\,\prob{\mathbf{G}_U=H}  
\end{equation}
where the infinitesimal $o(1)$ is the same for all $H$.

Consider independent random graphs $\bG_U$, $\bG_W$, and $\bG_{U\times W}$
and note that $\rg$ has the same distribution as the union $\bG_U\cup\bG_{U\times W}\cup\bG_W$.
We begin with estimation of the probability
\begin{eqnarray}
&&\cprobb{\mathbf{U}=U}{\rg|_{U}=H}=\cprob{U(\bG_U\cup\bG_{U\times W}\cup\bG_W)=U}{\bG_U=H}\label{eq:prob-H}\\
&&=\mathsf{P}[\forall x\in U\;\deg_{\bG_{U\times W}}(x)\equiv-\deg_H(x)\Mod{k},\nonumber\\
&&\mbox{}\hspace{14mm}\forall y\in W\;\deg_{\bG_W}(y)\not\equiv-\deg_{\bG_{U\times W}}(y)\Mod{k}].\nonumber
\end{eqnarray}
In the case that $\deg_{\bG_{U\times W}}(x)\equiv-\deg_H(x)\Mod{k}$ for every $x\in U$,
let us say that $\bG_{U\times W}$ is $H$-compatible. By Lemma \ref{lem:degrees_divisible},
this event happens with probability $k^{-u}(1+o(1/u))^u=k^{-u}(1+o(1))$.
Conditioned on the $H$-compatibility of $\bG_{U\times W}$, we have 
$$\deg_{\bG_W}(y)\not\equiv-\deg_{\bG_{U\times W}}(y)\Mod{k}$$
for all $y\in W$ with probability $(k-1)^wk^{-w}(1+o(1))$, which follows from Lemma \ref{lm:graphs_division}.
Therefore, the probability in \refeq{prob-H} is equal to $(k-1)^wk^{-n}(1+o(1))$.

Now, we have
$$
\probb{\rg|_{U}=H,\;\mathbf{U}=U}=\cprobb{\mathbf{U}=U}{\rg|_{U}=H}\probb{\rg|_{U}=H}=\frac{(k-1)^w}{k^n}\,(1+o(1))\probb{\bG_U=H}
$$
and also
$$
\prob{\mathbf{U}=U}=\sum_H\cprobb{\mathbf{U}=U}{\rg|_{U}=H}\prob{\rg|_{U}=H}=\frac{(k-1)^w}{k^n}\,(1+o(1)).
$$
Combining these two estimates, we immediately obtain the desired equality~\refeq{HU}.

\medskip

\noindent
\textbf{2.} 
We use the same proof strategy as in part 1.  
Fix an arbitrary graph $H$ with $V(H)=W\cup R$. It suffices to prove that
\begin{equation}\label{eq:HRW}
 \cprobb{\rg|_{W\cup R}=H}{\vec{\mathbf{U}}={\vec U}} = 
 (1+o(1))\,\probb{\mathbf{G}_{W\cup R}=H},
\end{equation}
where the infinitesimal $o(1)$ is the same for all $H$.

We represent $\rg$ as the union of three independent random graphs 
$$
\rg=\bG_{W\cup R}\cup\bG_{(W\cup R)\times U_0}\cup\bG_{U_0}
$$
and first estimate the probability $\cprob{\vec{\mathbf{U}}={\vec U}}{\bG_{W\cup R}=H}$.

We can assume that $\rg$ is exposed in three steps: first $\bG_{W\cup R}$ (which is supposed to be $H$),
then $\bG_{(W\cup R)\times U_0}$, and finally $\bG_{U_0}$.
Note that after the 2nd step, all neighbors of each vertex $x\in R$ and each vertex $y\in W$ are known
and we can now check whether or not $x\in U_r(\rg)$ for $r=1,\ldots,10$ and whether or not $y\in W(\rg)$.
If this agrees with the condition $\vec{\mathbf{U}}={\vec U}$, we say that $\bG_{(W\cup R)\times U_0}$ is $H$-compatible.
More specifically, $H$-compatibility means that 
\begin{equation}\label{eq:y}
\deg_{\bG_{(W\cup R)\times U_0}}(y)\not\equiv-\deg_H(y)\Mod{k}  
\end{equation}
for every $y\in W$ and that
\begin{equation}\label{eq:x}
\deg_{\bG_{(W\cup R)\times U_0}}(x)\equiv-\deg_H(x)\Mod{k}\text{ and }\deg_{\bG_{(W\cup R)\times U_0}}(x)\equiv r-\deg_{H|_R}(x)\Mod{11}.
\end{equation}
for every $r=1,\ldots,10$ and $x\in U_r$.
Applying Lemma \ref{lem:degrees_divisible}, we see that Condition \refeq{y} holds for all $y\in W$
with probability $\of{1-\frac1k(1+o(1/u_0))}^w=(1-\frac1k)^w(1+o(1))$.
Since $k$ and $11$ are coprime, the Chinese remainder theorem implies that Condition \refeq{x}
is equivalent to the relation
$$
\deg_{\bG_{(W\cup R)\times U_0}}(x)\equiv r_x\Mod{11k}
$$
for some integers $r_x$'s. By Lemma \ref{lem:degrees_divisible},
this is fulfilled for all $x\in R$ with probability $(11k)^{-u+u_0}(1+o(1/u_0))^{u-u_0}=(11k)^{-u+u_0}(1+o(1))$.
Thus, $\bG_{(W\cup R)\times U_0}$ is $H$-compatible with probability $(1-1/k)^w(11k)^{-u+u_0}(1+o(1))$.

Now, assume that $\bG_{(W\cup R)\times U_0}$ is $H$-compatible. After exposing $\bG_{U_0}$, 
we can decide whether $\vec{\mathbf{U}}={\vec U}$ by checking whether $z\in U_0(\rg)$ for every $z\in U_0$.
Under the condition $\bG_{W\cup R}\cup\bG_{(W\cup R)\times U_0}=H'$,
where $H'$ is fixed so that $H'|_{W\cup R}=H$, we have to check whether or not 
$$
\deg_{\bG_{U_0}}(z)\equiv-\deg_{H'}(z)\Mod{k}\text{ and }\deg_{\bG_{U_0}}(z)\equiv -\deg_{H'|_R}(x)\Mod{11},
$$
which, by the Chinese remainder theorem, is equivalent to
$$
\deg_{\bG_{U_0}}(z)\equiv r_z\Mod{11k}
$$
for some integers $r_z$'s. By Lemma \ref{lm:graphs_division}, this relation holds true for all $z\in U_0$
with probability $(11k)^{-u_0}(1+o(1))$.

Putting it all together, we conclude that
$$
\cprob{\vec{\mathbf{U}}={\vec U}}{\rg|_{W\cup R}=H}=\of{1-\frac1k}^w\of{\frac1{11k}}^u(1+o(1)).
$$
From this we derive
\begin{multline*}
\probb{\rg|_{W\cup R}=H,\;\vec{\mathbf{U}}={\vec U}}=\cprobb{\vec{\mathbf{U}}={\vec U}}{\rg|_{W\cup R}=H}\probb{\rg|_{W\cup R}=H}\\
=\of{1-\frac1k}^w\of{\frac1{11k}}^u(1+o(1))\,\probb{\bG_{W\cup R}=H}  
\end{multline*}
and
$$
\prob{\vec{\mathbf{U}}={\vec U}}=\sum_H\cprobb{\vec{\mathbf{U}}={\vec U}}{\rg|_{W\cup R}=H}\prob{\rg|_{W\cup R}=H}
=\of{1-\frac1k}^w\of{\frac1{11k}}^u(1+o(1)).
$$
Combining the two estimates, we arrive at the desired equality~\refeq{HRW}.
\end{proof}

\subsection{Completing the proof of Theorem~\ref{thm:unlab}}
\label{sec:th1_proof_finish}

We split the proof in three parts according to the outline in Subsection~\ref{ss:prop}.
 
\subsubsection{The probability of $Q_k$ is small}
\label{sec:th1_proof_finish_probQ_small}

Our nearest goal is to show that $\rg\in B$ whp.
The following fact is the first step in this direction.
We call a set $U\subset[n]$ \emph{standard} if
$$
\frac{n}{k}-\sqrt{n}\ln n \le |U|\le \frac{n}{k}+\sqrt{n}\ln n.
$$
Recall that $\mathbf{U}$ consists of all vertices in $\rg$ with degrees divisible by~$k$.

\begin{claim}\label{cl:U}
$\mathbf{U}$ is standard with probability $1-o(1/n^3)$.  
\end{claim}

\begin{proof}
For a vertex $i\in[n]$, let $\xi_i$ be the indicator variable for the event that $\deg_{\rg}(i)$ is divisible by $k$.
Note that $|\mathbf{U}|=\xi_1+\ldots+\xi_n$. 
Let a random variable $\mathbf{p}$ and $n$ random variables $\eta_1(x),\ldots,\eta_n(x)$, for each $x\in(0,1)$,
be as in Theorem~\ref{th:MW}.
Define $X_i(x)$ to be the indicator variable for the event that $\eta_i(x)$ is divisible by $k$
and set $X(x)=X_1(x)+\ldots+X_n(x)$. Since this sum is a counterpart of the sum $\xi_1+\ldots+\xi_n$,
we conclude by Theorem~\ref{th:MW} that
\begin{multline}
\probb{\left||\mathbf{U}|-\frac{n}{k}\right|>\sqrt{n}\ln n} \le \\
(1+o(1))\cprobb{\left|X(\mathbf{p})-\frac{n}{k}\right|>
\sqrt{n}\ln n }{ \sum_{i=1}^{n}\eta_i(\mathbf{p})\text{ is even} }+o\left(\frac{1}{n^3}\right).\label{eq:to-MW}
\end{multline}
It suffices, therefore, to prove that the conditional probability above is~$o(1/n^3)$.

Let $f$ be the density of $\mathbf{p}$, and $\Phi(x)=\int_{-\infty}^x\frac{1}{\sqrt{2\pi}}e^{-t^2/2}dt$ be the cumulative 
distribution function of a standard normal random variable. Then
\begin{multline}
\probb{|\mathbf{p}-1/2|>\frac{2\ln n}{n}}=\int_0^{1/2-\frac{2\ln n}{n}}f(x)dx+\int_{1/2+\frac{2\ln n}{n}}^1f(x)dx\\
=\frac{1}{1 - 2\Phi(-\sqrt{n(n-1)})}\left[\int_{-\sqrt{n(n-1)}}^{-4\ln n\sqrt{1-1/n}}+\int_{4\ln n\sqrt{1-1/n}}^{\sqrt{n(n-1)}}\right]
\frac{1}{\sqrt{2\pi}}e^{-x^2/2}dx=o\left(\frac{1}{n^3}\right)
\label{eq:p_is_near_1/2} 
\end{multline}
since $1-\Phi(x)\sim \frac{1}{\sqrt{2\pi} x}e^{-x^2/2}$ as $x\to\infty$; see, e.g.,~\cite{Feller}. 

Since $\sum_{i=1}^{n}\eta_i(x)$ has binomial distribution with parameters $n(n-1)$ and $x$, 
part 1 of Lemma \ref{lem:character} implies that
\begin{equation}
 \probb{\sum_{i=1}^{n}\eta_i(x)\text{ is even}}=\frac{1}{2}+\frac{1}{2}(1-2x)^{n(n-1)}\geq\frac{1}{2}
\label{sum_is_even}
\end{equation}
for every $x\in(0,1)$. It follows that
$$
\cprobb{\left|X(\mathbf{p})-\frac{n}{k}\right| > \sqrt{n}\ln n }{ \sum_{i=1}^{n}\eta_i(\mathbf{p})\text{ is even} } \le
2\,\probb{\left|X(\mathbf{p})-\frac{n}{k}\right| > \sqrt{n}\ln n }.
$$
Furthermore,
\begin{multline} 
\probb{\left|X(\mathbf{p})-\frac{n}{k}\right|>\sqrt{n}\ln n}=
  \int_0^1\probb{\left|X(x)-\frac{n}{k}\right|>\sqrt{n}\ln n }f(x)\,dx\\
\le
\int_{1/2-2\ln n/n}^{1/2+2\ln n/n}\probb{\left|X(x)-\frac{n}{k}\right|>\sqrt{n}\ln n }f(x)dx
+\probb{|\mathbf{p}-1/2|>\frac{2\ln n}{n}}\\
\le
\max_{|x-1/2|\leq 2\ln n/n}\probb{X(x)>\frac{n}{k}+\sqrt{n}\ln n}+
\max_{|x-1/2|\leq 2\ln n/n}\probb{X(x)<\frac{n}{k}-\sqrt{n}\ln n}+o\left(\frac{1}{n^3}\right),\label{eq:Xx}
\end{multline} 
where we make use of Estimate \refeq{p_is_near_1/2}.
Lemma~\ref{lem:degrees_divisible} shows that if $|x-1/2|\leq\frac{2\ln n}{n}$,
then we have $\prob{X_i(x)=1}=\frac{1}{k}+o(1/n)$.  
Applying the Chernoff bound, we see that both maximums in \refeq{Xx} are $o(1/n^3)$.
Putting it together, we conclude that
$$
\cprobb{\left|X(\mathbf{p})-\frac{n}{k}\right| > \sqrt{n}\ln n }{ \sum_{i=1}^{n}\eta_i(x)\text{ is even} }=o\left(\frac{1}{n^3}\right),
$$
and the claim follows from~\refeq{to-MW}.
\end{proof}

Fix a standard $U\subset[n]$ and
let ${\vec U}=(U_0,\ldots,U_{10})$ be a partition of $U$. We call ${\vec U}$ \emph{standard} if
$$
\frac{|U|}{11}-\sqrt{n}\ln n \le |U_r|\le \frac{|U|}{11}+\sqrt{n}\ln n
$$
for every $r=0,\ldots,10$.
Using virtually the same argument as in the proof of Claim \ref{cl:U}, we have
$\prob{|U_r(\bG_U)-|U|/11|>\sqrt{n}\ln n} =o(1/n^3)$ for every $r=0,\ldots,10$.
In other words, $(U_0(\bG_U),\ldots,U_{10}(\bG_U))$ is standard whp.
Using part 1 of Lemma \ref{lm:asympt_uniform}, we conclude that
\begin{equation}
\cprobb{\vec{\mathbf{U}}\text{ is standard}}{\mathbf{U}=U}=
1-o\left(\frac{1}{n^3}\right).
\label{eq:U_i_concentration}
\end{equation}

Assume that ${\vec U}$ is standard and consider a random graph $\bG_{W\cup R}$ where, as usually, $W=[n]\setminus U$ and
$R=U\setminus U_0$. As follows from Estimate (\ref{binomial_bound}),
two fixed vertices $u,v\in W$ of $\bG_{W\cup R}$ have equally many neighbors in $U_r$ for every $r=1,\ldots,10$
with probability $O(1/n^5)$. By the union bound, the partition $R=U_1\cup\ldots\cup U_{10}$
does not resolve $W$ in $\bG_{W\cup R}$ with probability $O(1/n^3)$. By part 2 of Lemma~\ref{lm:asympt_uniform}
we conclude that, conditioned on $\vec{\mathbf{U}}={\vec U}$, $R(\rg)$ does not resolve $W(\rg)$ with asymptotically the same probability, that is,
\begin{equation}
\cprobb{\overline{B} }{ \vec{\mathbf{U}}={\vec U} } = O\of{\frac1{n^3}},
\label{eq:distinct_degree_vectors}
\end{equation}
where $\overline{B}$ denotes the event $\rg\notin B$.

We now can see that the event $B$ holds with high probability. Indeed,
taking into account Estimate \refeq{distinct_degree_vectors}, Claim \ref{cl:U}, and Estimate \refeq{U_i_concentration},
we have
\begin{multline*}
\probb{\overline{B}}\leq
\sum_{U\text{ and }{\vec U}\text{ standard}}
\cprobb{\overline{B} }{ \vec{\mathbf{U}}={\vec U} }\prob{\vec{\mathbf{U}}={\vec U}} + 
    \probb{\mathbf{U}\text{ is not standard}}\\
+\sum_{U\text{ standard}}
\cprobb{\vec{\mathbf{U}}\text{ is not standard} }{ \mathbf{U}=U }
         \prob{\mathbf{U}=U}=O\left(\frac{1}{n^3}\right).
\end{multline*}

With this upper bound for the probability of $\overline{B}$, we are ready to estimate the probability
of $Q_k$ from above. Note that
$$
\prob{Q_k} = \prob{Q_k\cap B} + \prob{\overline{B}} = \prob{Q_k\cap B}+o(1/n^2).
$$
By Claim \ref{cl:U} and Estimate \refeq{U_i_concentration},
$$
\prob{Q_k\cap B}=\sum_{U\text{ and }{\vec U}\text{ standard}}
\cprob{Q_k\cap B}{\vec{\mathbf{U}}={\vec U}}\,\prob{\vec{\mathbf{U}}={\vec U}}+o\left(\frac{1}{n^3}\right).
$$
By part 2 of Lemma~\ref{lm:asympt_uniform}, the probability $\cprob{Q_k\cap B}{\vec{\mathbf{U}}={\vec U}}$
is asymptotically the same as the probability that in the uniformly distributed random graph $\bG_{W\cup R}$,
simultaneously,
\begin{enumerate}[(1)]
\item 
$R=U_1\cup\ldots\cup U_{10}$ resolves $W=[n]\setminus U$, and
\item 
the canonically relabeled subgraph $\bG_{W\cup R}|_W$ belongs to $Q$.
\end{enumerate}
Assume that the former condition is fulfilled. Since the random graphs $\bG_{W\cup R}|_W=\bG_W$
and $\bG_{W\times R}$ are independent, the latter condition has the same probability as the event
$\bG_{|W|}\in Q$, which does not exceed $\frac{2+1/k+o(1)}{(n-n/k-\sqrt{n}\ln n)^2}$.
We conclude that
\begin{equation}\label{eq:Q1}
 \prob{\rg\in Q_k} \le \frac{(2+1/k)(1+o(1))}{(1-1/k)^2n^2} + o\left(\frac{1}{n^2}\right)\le 
\of{1+\frac3k}\,\frac2{n^2}, 
\end{equation}
where the last inequality is fulfilled for all sufficiently large~$n$.

\subsubsection{$Q_k$ is almost anti-stochastic}

For a graph $G$ and two distinct vertices $u,v\in V(G)$, let $G(u,v)$ denote the graph obtained from $G$ 
by changing the adjacency between $u$ and $v$. If $u=v$, we set $G(u,v)=G$.

Let $A$ denote the event
$$
\exists\, u,v\in\mathbf{W}\quad\rg(u,v)\in Q_k.
$$
It is enough to prove that $A$ has high probability.

In what follows, for a partition ${\vec U}=(U_0,\ldots,U_{10})$ of $U$, the event $\vec{\mathbf{U}}={\vec U}$
will for brevity be denoted by $C_{\vec U}$. With some abuse of notation, we write $A$ and $C_{\vec U}$
also to denote the corresponding \emph{sets} of graphs. We have
\begin{multline*}
 \prob{A}  \geq \prob{A\cap B}\\
 \geq\sum_{U\text{ and }{\vec U}\text{ standard}}
 \cprob{A\cap B}{C_{\vec U}}\,\prob{C_{\vec U}} 
=\sum_{U\text{ and }{\vec U}\text{ standard}}
 \cprob{A}{B\cap C_{\vec U}}\,\cprob{B}{C_{\vec U}}\,\prob{C_{\vec U}}.
\end{multline*}
We can bound $\cprob{B}{C_{\vec U}}$ from below according to Estimate \refeq{distinct_degree_vectors}.
The probability $\cprob{A}{B\cap C_{\vec U}}$ can also be bounded according to Claim \ref{cl:ABC} below.
We, therefore, obtain
$$
 \prob{A}  \ge 
\of{1-o\of{\frac1{n^2}}}
\of{1-\frac{4}{k}+\frac{4}{k^2}}(1-o(1))
\sum_{U\text{ and }{\vec U}\text{ standard}}\prob{C_{\vec U}}.
$$
Note that
\begin{multline*}
\sum_{U\text{ and }{\vec U}\text{ standard}}\prob{C_{\vec U}}=\prob{\mathbf{U}\text{ and }{\vec{\mathbf{U}}}\text{ are standard}}\\
=
\cprob{{\vec{\mathbf{U}}}\text{ is standard}}{\mathbf{U}\text{ is standard}}\,\prob{\mathbf{U}\text{ is standard}}
\ge 1-o\of{\frac1{n^3}},
\end{multline*}
where the last inequality follows from Estimate \refeq{U_i_concentration} and Claim \ref{cl:U}.
We conclude that 
\begin{equation}\label{eq:Q2}
 \prob{A}  \ge 1-\frac{4}{k}\text{ for sufficiently large }n.  
\end{equation}

It remains to prove the estimate of $\cprob{A}{B\cap C_{\vec U}}$ we used above.

\begin{claim}\label{cl:ABC}
If a set $U\subset[n]$ and its partition ${\vec U}=(U_0,\ldots,U_{10})$ are standard, then
$$
\cprob{A}{B\cap C_{\vec U}}\ge\of{1-\frac{4}{k}+\frac{4}{k^2}}(1-o(1)).
$$
\end{claim}

\begin{proof}
As usually, let $W=[n]\setminus U$ and $R=U\setminus U_0$.
A graph $G$ with $V(G)=[n]$ is the union of six edge-disjoint subgraphs, namely
$$
G=G_W\cup G_R\cup G_{U_0}\cup G_{R\times W}\cup G_{R\times U_0}\cup G_{W\times U_0},
$$
where $G_X$ is a graph on the vertex set $X$ and $G_{X\times Y}$ is a bipartite graph with vertex classes $X$ and $Y$.
Suppose that the subgraph $G_R\cup G_{R\times W}\cup G_{R\times U_0}$ is fixed. We denote this assignment by $\alpha$
and write $G_\alpha=G_R\cup G_{R\times W}\cup G_{R\times U_0}$ for brevity. Moreover, in what follows we consider only
those assignments $\alpha$ for which the following three conditions are met:
\begin{itemize}
\item 
$\deg_{G_\alpha}(z)\equiv 0\Mod{k}$ for every $z\in R$.
\item 
$\deg_{G_\alpha|_U}(z)\equiv r\Mod{11}$ for every $r=1,\ldots,10$ and $z\in U_r$.
\item 
The partition $R=U_1\cup\ldots\cup U_{10}$ resolves $W$ in $G_\alpha$
and, hence, defines a relabeling of~$W$.
\end{itemize}
Obviously, these conditions will hold true for any extension of $G_\alpha$ by $G_W\cup G_{W\times U_0}\cup G_{U_0}$.

Let $H$ be a graph with $V(H)=W$. For a fixed $\alpha$, we consider only those $H$ for which
the graph $H\cup G_\alpha$ has the following property: 
the relabeled subgraph $H$ of $H\cup G_\alpha$ belongs to $Q$.
For $u,v\in V(H)$, we define two classes of graphs on $[n]$:
\begin{itemize}
\item 
$\mathcal{H}_\alpha(H,u,v)$ consists of all $G$ such that $G_R\cup G_{R\times W}\cup G_{R\times U_0}=G_\alpha$, $G_W=H(u,v)$, and $G\in C_{\vec U}$.
\item 
$\mathcal{H}_\alpha^+(H,u,v)$ consists of all $G$ in $\mathcal{H}_\alpha(H,u,v)$
with additional property that also $G(u,v)\in C_{\vec U}$.
\end{itemize}
Since $Q$ is defined from a covering code, 
\begin{equation}\label{eq:BCU}
B\cap C_{\vec U}=\bigcup_\alpha\bigcup_{H,u,v}\mathcal{H}_\alpha(H,u,v).
\end{equation}
Here and below, the union is taken over all $\alpha$ and $H$ possessing the properties specified  above.
It is also important to note that
\begin{equation}\label{eq:ABCU}
A\cap B\cap C_{\vec U}\supseteq\bigcup_\alpha\bigcup_{H,u,v}\mathcal{H}_\alpha^+(H,u,v).
\end{equation}
Therefore, the claim will be proved if we show that, for each $\alpha$,
the set $\bigcup_{H,u,v}\mathcal{H}_\alpha^+(H,u,v)$ is a large fraction of the set $\bigcup_{H,u,v}\mathcal{H}_\alpha(H,u,v)$.

To this end, we first prove that
\begin{equation}\label{eq:HH}
  |\mathcal{H}_\alpha^+(H,u,v)|\geq\frac{(1-2/k)^2(1+o(1))}{(1-1/k)^2}\,|\mathcal{H}_\alpha(H,u,v)|,
\end{equation}
where the infinitesimal $o(1)$ is the same for all $H$, $u$ and $v$.
This is obvious if $u=v$. Suppose that $u\ne v$ and
consider random graphs $\bG'=\bG_{W\times U_0}\cup H(u,v)\cup G_\alpha$ and $\bG''=\bG'\cup\bG_{U_0}$,
where $\bG_{W\times U_0}$ and $\bG_{U_0}$ are independent. Note that $\bG''\in\mathcal{H}_\alpha(H,u,v)$
exactly when
\begin{equation}\label{eq:yyy}
\deg_{\bG_{W\times U_0}}(y)\not\equiv-\deg_{H(u,v)\cup G_\alpha}(y)\Mod{k}
\end{equation}
for every $y\in W$ and
\begin{equation}\label{eq:xxx}
\deg_{\bG_{U_0}}(x)\equiv-\deg_{\bG'}(x)\Mod{k}\text{ and }\deg_{\bG_{U_0}}(x)\equiv -\deg_{G_\alpha}(x)\Mod{11}.
\end{equation}
for every $x\in U_0$.
By Lemma \ref{lem:degrees_divisible}, Condition \refeq{yyy} is fulfilled for all $y\in W$
with probability $\of{1-\frac1k(1+o(1/u_0))}^w=(1-\frac1k)^w(1+o(1))$, where $u_0=|U_0|$ and $w=|W|$.
Since $k$ and $11$ are coprime, the Chinese remainder theorem allows us to express Condition \refeq{xxx} more concisely.
Under the condition that $\bG'=G'$ for a fixed graph $G'$, Condition \refeq{xxx} is equivalent to the relation
$$
\deg_{\bG_{U_0}}(x)\equiv r_x\Mod{11k}
$$
for some integers $r_x$'s. By Lemma \ref{lm:graphs_division}, this is fulfilled for all $x\in U_0$ with probability
$(11k)^{-u_0}(1+o(1/u_0))^{u_0}=(11k)^{-u_0}(1+o(1))$. It follows that
\begin{equation}\label{eq:proba-H}
 \prob{\bG''\in\mathcal{H}_\alpha(H,u,v)}=\of{1-\frac1k}^w(11k)^{-u_0}(1+o(1)). 
\end{equation}
Similarly, we have
\begin{multline}
\prob{\bG''\in\mathcal{H}_\alpha^+(H,u,v)}=
\prob{\bG'',\bG''(u,v)\in\mathcal{H}_\alpha(H,u,v)}\\=
\of{1-\frac2k}^2\of{1-\frac1k}^{w-2}(11k)^{-u_0}(1+o(1)).\label{eq:proba-H+}
\end{multline}
The only difference in the derivation of this estimate is that, in addition to \refeq{yyy},
for $y\in\{u,v\}$ we also have to take into account the condition
$$
\deg_{\bG_{W\times U_0}}(y)\not\equiv-\deg_{H\cup G_\alpha}(y)\Mod{k}.
$$
Inequality \refeq{HH} now follows by considering the quotient of the probabilities in \refeq{proba-H+} and~\refeq{proba-H}.

Since $Q$ is asymptotically optimal and Estimate \refeq{proba-H} is uniform over $H$, $u$, and $v$, we have
$$
 \sum_{H,u,v}|\mathcal{H}_\alpha(H,u,v)|=(1+o(1))\left|\bigcup_{H,u,v}\mathcal{H}_\alpha(H,u,v)\right|,
$$ 
where here and below the sum goes over unordered pairs $u,v$.
Using this along with Inequality~\refeq{HH}, we obtain
\begin{align*}
 \left|\bigcup_{H,u,v}\mathcal{H}_\alpha^+(H,u,v)\right| & 
\geq\left|\bigcup_{H,u,v}\mathcal{H}_\alpha(H,u,v)\right|-\sum_{H,u,v}|\mathcal{H}_\alpha(H,u,v)\setminus\mathcal{H}_\alpha^+(H,u,v)|\\
& =\left|\bigcup_{H,u,v}\mathcal{H}_\alpha(H,u,v)\right|-\sum_{H,u,v}\biggl(|\mathcal{H}_\alpha(H,u,v)|-|\mathcal{H}_\alpha^+(H,u,v)|\biggr)\\
&\geq\left|\bigcup_{H,u,v}\mathcal{H}_\alpha(H,u,v)\right|-
\left(
1-\frac{(1-2/k)^2(1+o(1))}{(1-1/k)^2}
\right)
\sum_{H,u,v}|\mathcal{H}_\alpha(H,u,v)|\\
&\ge\left(1-\frac{4}{k}+\frac{4}{k^2}\right)(1+o(1))\left|\bigcup_{H,u,v}\mathcal{H}_\alpha(H,u,v)\right|.
 \end{align*}

Since $\cprob{A}{B\cap C_{\vec U}}=|A\cap B\cap C_{\vec U}|/|B\cap C_{\vec U}|$,
the claim follows from Equality \refeq{BCU} and Inclusion \refeq{ABCU} by taking into account
that the sets $\mathcal{H}_\alpha(H,u,v)$ for different $\alpha$ are disjoint 
(and the same is hence true also for the sets $\mathcal{H}_\alpha^+(H,u,v)$).
\end{proof}

\subsubsection{Merging all $Q_k$'s together}

It remains to convert the sequence of graph properties $Q_k$ into a single anti-stochastic property $Q^*$. 
Based on Estimates \refeq{Q1} and \refeq{Q2}, for each $Q_k$ we define an integer $N_k$ such that
$N_k>N_{k'}$ if $k'<k$ and the inequalities $\prob{\rg\in Q_k}\le(1+3/k)\,\frac2{n^2}$ and
$\prob{\exists u,v\;\rg(u,v)\in Q_k}  \ge 1-4/k$ are true for all $n\ge N_k$.
Let $k(n)$ be the maximum $k$ such that $N_k\le n$. Define $Q^*$ to be the event that $\rg\in Q_{k(n)}$.
The graph property $Q^*$ is anti-stochastic because $k(n)\to\infty$ as $n\to\infty$.
The proof of Theorem \ref{thm:unlab} is complete.

\begin{remark}\label{rem:it-degrees}
  We define the \emph{$s$-iterated degree} $d_s(x)$ of a vertex $x$ recursively by setting
  $d_0(x)=0$ and $d_s(x)=\Mset{d_{s-1}(y)}{y\in N(x)}$, where $\Mset{}$ denotes a multiset
  and $N(x)$ stands for the neighborhood of $x$. Obviously, $d_1(x)$ represents exactly
  the degree of $x$. Note that if $d_s(x)=d_s(x')$, then $d_t(x)=d_t(x')$ for all $t\le s$.

  We say that graphs $G$ and $G'$ have the
  same $s$-iterated degree sequence if the multisets of their $s$-iterated degrees are equal or,
  equivalently, if there is a bijection $f$ from the vertex set of $G$ onto the vertex set of $G'$
  such that $d_s(x)=d_s(f(x))$ for every vertex $x$ of $G$. We write $G\equiv_sG'$ in this case.
  Note that $\equiv_s$ is an equivalence relation coarser than graph isomorphism. We claim that
  the anti-stochastic property $Q^*$ constructed in the proof of Theorem \ref{thm:unlab} is
  not only isomorphism invariant but even $\equiv_4$-invariant, that is, if $G\equiv_4G'$
  and $G\in Q^*$, then $G'\in Q^*$ too.

  Indeed, let $f$ be a bijection preserving $4$-iterated degrees. Since $f$ preserves vertex degrees,
  we have $U(G')=f(U(G))$ for each parameter $k$. Since $f$ preserves twice-iterated degrees, we have
  $U_r(G')=f(U_r(G))$ for all $0\le r\le 10$. Since $f$ preserves $3$-iterated degrees, we have $G'\notin B$
  whenever $G\notin B$. Suppose that both $G$ and $G'$ possess the property $B$.
  This implies that every vertex in $W(G)$ and $W(G')$ has a unique
  $3$-iterated degree. The restriction of $f$ on $W$ is a bijection from $W(G)$ to $W(G')$ and matches
  the 3-iterated degrees. In other words, $f$ respects the canonical labeling of $G|_{W(G)}$ and $G'|_{W(G')}$
  defined in the proof. Combining this fact with our assumption that $f$ preserves $4$-iterated degrees,
  we conclude that $f$ is an isomorphism from $G|_{W(G)}$ to $G'|_{W(G')}$. Therefore, $G'\in Q^*$
  whenever $G\in Q^*$, as claimed.
\end{remark}

\section{Proof of Theorem \ref{thm:degrees}}

\subsection{Degree sequences}\label{ss:degrees}

Let $G$ be a graph with $V(G)=[n]$.
While the degree sequence of this graph is the \emph{multiset} of its vertex degrees,
we define the \emph{labeled degree sequence} of $G$ as the \emph{sequence} $(\deg_G(1),\ldots,\deg_G(n))$.
In Section \ref{ss:deg-seq-lower} we will need a result of Liebenau and Wormald~\cite{LiebenauWormald} on 
enumeration of graphs with a given labeled degree sequence.
Note that this result appeared earlier in a weaker form in~\cite{McKayW90}, which would also be sufficient for our purposes.

Given a sequence of non-negative integers $\seq{d}=(d_1,\ldots,d_n)$, let $g(\seq{d})$ denote the number 
of graphs on $[n]$ whose labeled degree sequence is $\seq{d}$. Let 
$$
d=\frac{1}{n}\sum_{i=1}^nd_i,\quad
\mu=d/(n-1),\quad\text{and}\quad
\gamma=(n-1)^{-2}\sum_{i=1}^n(d_i-d)^2.
$$

\begin{theorem}[Liebenau and Wormald \cite{LiebenauWormald}]
There is $\varepsilon>0$ such that if 
$$
\max_i|d_i-d|=o\left(n^{\varepsilon}\sqrt{\min\{d,n-d-1\}}\right),\quad n\min\{d,n-d-1\}\to\infty,
$$ 
and $\sum_{i=1}^n d_i$ is even, then
$$
 g(\seq{d})=(\sqrt{2}+o(1))\exp\left(\frac{1}{4}-\frac{\gamma^2}{4\mu^2(1-\mu)^2}\right)\left(\mu^{\mu}(1-\mu)^{1-\mu}\right)^{n(n-1)/2}\prod_{i=1}^n{n-1\choose d_i}.
$$
where the infinitesimal $o(1)$ is the same for all $\seq{d}$.
\label{th:counting_graps_deg_sequence}
\end{theorem}

We also need some technical facts on the frequency of a particular vertex degree in the degree sequence
of the random graph.
For an integer $y$, let $N_y$ be the number of vertices of degree $y$ in $\rg$.
Let 
$$
\zeta_1=\frac{\ln(2\pi)+\frac32\ln\ln n}{\ln n},\qquad 
\zeta_2=\frac{\ln(2\pi)+\frac52\ln\ln n}{\ln n}.
$$
The first four parts of the following lemma will be used in Section \ref{ss:deg-seq-lower},
while the last part is needed in Section~\ref{sec41}.

\begin{lemma} 
For $i\in[n]$, let $\xi_i=\deg_{\rg}(i)$.
The following conditions are true whp.
\begin{enumerate}[\bf 1.]
    \item For every $i$, 
    $$
    \left|\xi_i-\frac{n}{2}\right|<\sqrt{\frac{n\ln n}{2}}.
    $$
    \item There are at most $4\sqrt{n}(\ln n)^{1/4}$ vertices $i$ such that  
    $$
    \left|\xi_i-\frac{n}{2}\right|>\frac{1}{2}\sqrt{n\ln n\left(1-\zeta_1\right)}.
    $$
    \item 
      \begin{enumerate}[(a)]
      \item 
If $|y-n/2| < \frac{1}{2}\sqrt{n\ln n\left(1-\zeta_2\right)}+1$, then
    $$
    N_y\geq(\ln n)^{5/4}\,\,\text{ and }\,\,
    |N_y-N_{y+1}|\leq\frac{N_y}{\ln\ln n}.
    $$
      \item 
    If $\frac{1}{2}\sqrt{n\ln n\left(1-\zeta_2\right)}-2\leq|y-n/2|\leq\frac{1}{2}\sqrt{n\ln n\left(1-\zeta_1\right)}+2$, then
    $N_y\leq(\ln n)^2$.
      \end{enumerate}
    \item For all but at most $\sqrt{n}$ vertices $i$, the condition
    $$
    \frac{1}{2}\sqrt{n\ln n\left(1-\zeta_2\right)}\leq\left|\xi_i-\frac{n}{2}\right|\leq\frac{1}{2}\sqrt{n\ln n\left(1-\zeta_1\right)}
    $$
    implies that
    $$
    N_y\ge(\ln n)^{3/4}\,\,\text{ as well as }\,\,
    |N_y-N_{y+1}|\le\frac{N_y}{\ln\ln n}\,\,\text{ and }\,\,
    |N_y-N_{y-1}|\le\frac{N_y}{\ln\ln n}
    $$
for every of three values $y\in\{\xi_i-1,\xi_i,\xi_i+1\}$.
    \item For every two, not necessarily distinct integers $y_1,y_2$ such that 
    $$
    |y_j-n/2|\leq\frac{1}{2}\sqrt{n\left(\ln n-2\sqrt{\ln n}\right)}+1,\quad j=1,2,
    $$
    there exists a pair of adjacent vertices of degrees $y_1,y_2$ and a pair of non-adjacent vertices of degrees $y_1,y_2$.
\end{enumerate}
\label{lem:degrees_range}
\end{lemma}

The first part is proven in~\cite[Corollary 7]{BollobasDegrees}. The proofs of parts 2, 3 and 4 are quite technical 
and based on the standard apllications of de Moivre--Laplace limit theorem, concentration inequalities and Theorem~\ref{th:MW}. 
Part 5 is also standard and follows from the fact that $\rg$ whp contains no large cliques, large independend sets, 
complete bipartite graphs with large parts, nor induced disconnected subgraphs with large components. We, therefore, move 
the proof of Lemma \ref{lem:degrees_range} to the appendix.

\subsection{Upper bound}\label{sec41}

Let $[a,b]$ denote the interval of integers $a,a+1,\ldots,b$.
We consider the interval of integers
$$
D_n=\Set{d\in\mathbb{Z}}{|d-n/2|\leq\frac{1}{2}\sqrt{n\left(\ln n-2\sqrt{\ln n}\right)}},
$$
appearing in part 5 of Lemma \ref{lem:degrees_range}.
The smallest and the largest integers in $D_n$ are denoted by $d_*$ and $d^*$ respectively;
thus, $D_n=[d_*,d^*]$. The integer $\lfloor n/2\rfloor$ splits $D_n$ in two parts
$$
D_{n,1}=[d_*,\lfloor n/2\rfloor-1]\text{ and }
D_{n,2}=[\lfloor n/2\rfloor+1,d^*],
$$
that is, 
$$
D_n=D_{n,1}\cup\Set{\lfloor n/2\rfloor}\cup D_{n,2}.
$$
Set 
$$
\delta=|D_n|,\quad\delta_1=|D_{n,1}|,\quad\text{and}\quad\delta_2=|D_{n,2}|.
$$
Shifting $D_{n,1}$ and $D_{n,2}$ in 1, we obtain the intervals
$$
D_{n,1}^+=[d_*+1,\lfloor n/2\rfloor]\text{ and }
D_{n,2}^+=[\lfloor n/2\rfloor+2,d^*+1].
$$

Recall that, for a non-negative integer $y$, the number of vertices of degree $y$ in $\rg$ is denoted by $N_y$.
For an integer $y>d_*$, we define $X_y$ to be the total number of vertices $v$ such that $d_*\le\deg(v)\le y$,
that is, $X_y=\sum_{d=d_*}^y N_d$.

\subsubsection{The property}

We define two integer sequences
$$
\mathbf{X}^{\downarrow}=(X_y)_{y\in D_n^1}\text{ and }
\mathbf{X}^{\uparrow}=(X_y)_{y\in D_n^2}
$$
of length $\delta_1$ and $\delta_2$ respectively.
For further reference, we make a simple observation.

\begin{claim}\label{cl:add-del}
Let $u$ and $v$ be vertices of degrees $\deg(u)=x$ and $\deg(v)=y$.
\begin{enumerate}[\bf 1.]
\item 
If $x\in D_{n,1}$, $y\in D_{n,2}$, and $u$ and $v$ are non-adjacent, then addition of an edge between $u$ and $v$ 
changes only one coordinate of $\mathbf{X}^{\downarrow}$, namely the one indexed by $x$, and only one coordinate 
of $\mathbf{X}^{\uparrow}$, namely the one indexed by $y$. Both coordinates decrease by~1.
\item 
If $x\in D_{n,1}^+$, $y\in D_{n,2}^+$, and $u$ and $v$ are adjacent, then deletion of the edge between $u$ and $v$ 
changes only one coordinate of $\mathbf{X}^{\downarrow}$, namely the one indexed by $x-1$, and only one coordinate 
of $\mathbf{X}^{\uparrow}$, namely the one indexed by $y-1$. Both coordinates increase by~1.
\end{enumerate}
\end{claim}

Let $\mathbf{Y}^{\downarrow}\in\{0,1\}^{\delta_1}$ be the vector of parities of $\mathbf{X}^{\downarrow}$, that is, 
$$
(\mathbf{Y}^{\downarrow})_i=(\mathbf{X}^{\downarrow})_{d_*+i-1}\bmod2.
$$
Similarly, $\mathbf{Y}^{\uparrow}\in\{0,1\}^{\delta_2}$ is the vector of parities of $\mathbf{X}^{\uparrow}$. 
Let $\mathcal{S}^{\downarrow}\subset\{0,1\}^{\delta_1}$ and $\mathcal{S}^{\uparrow}\subset\{0,1\}^{\delta_2}$ be covering codes
with asymptotically optimal densities. Define
\begin{equation}
  \label{eq:Z}
\mathbf{Z}=\sum_{d_*\leq y\leq n/2} yN_y,  
\end{equation}
where the summation goes over all $y\in D_{n,1}\cup\Set{\lfloor n/2\rfloor}$.
Let $A$ be the property that
$$
\mathbf{Y}^{\downarrow}\in\mathcal{S}^{\downarrow},\quad
\mathbf{Y}^{\uparrow}\in\mathcal{S}^{\uparrow},\quad\text{and}\quad
\mathbf{Z}\bmod4\in\Set{0,1}.
$$
We show that this is the desired anti-stochastic property. 
The proof of the following equalities is postponed to the next subsection.

\begin{lemma}\label{lem:pair_of_degrees}\hfill
  \begin{enumerate}[\bf 1.]
  \item 
$\prob{\mathbf{Y}^{\downarrow}\in\mathcal{S}^{\downarrow}}=\frac{2+o(1)}{\sqrt{n\ln n}}$ and 
$\prob{\mathbf{Y}^{\uparrow}\in\mathcal{S}^{\uparrow}}=\frac{2+o(1)}{\sqrt{n\ln n}}$. 
\item 
$\prob{\rg\in A}=\frac{2+o(1)}{n\ln n}$. 
  \end{enumerate}
\end{lemma}

As the probability of $A$ is asymptotically determined by part 2 of Lemma \ref{lem:pair_of_degrees},
it remains to prove the existence of an adversary. 
A characteristic $\mathbf{V}$ of a random graph $\rg$ is, formally, a function defined on the set of graphs with vertices $1,\ldots,n$.
For a graph $G$ on this vertex set, we therefore write $\mathbf{V}(G)$ to denote the value of $\mathbf{V}$ on $G$.
Let $G$ be a graph on $[n]$ such that 
$\mathbf{Y}^{\downarrow}(G)\notin\mathcal{S}^{\downarrow}$,
$\mathbf{Y}^{\uparrow}(G)\notin\mathcal{S}^{\uparrow}$, and for every two integers $x,y\in[d_*,d^*+1]$ 
there exist a pair of adjacent vertices of degrees $x$ and $y$ and a pair of non-adjacent vertices of degrees $x$ and $y$. 
Note that these conditions hold whp for $\rg$ by part 1 of Lemma \ref{lem:pair_of_degrees} and part 5 of Lemma~\ref{lem:degrees_range}. 
Therefore, it suffices to show that, by adding or deleting one edge, the adversary can transform $G$ into a graph $G'$
possessing the property~$A$.

The adversary is going either to add an edge between two vertices $u$ and $v$ as in part 1 of Claim \ref{cl:add-del}
or to delete an edge between two vertices $u$ and $v$ as in part 2 of this claim.
Note that $\mathbf{Z}(G')=\mathbf{Z}(G)+1$ in the case of addition, and
$\mathbf{Z}(G')=\mathbf{Z}(G)-1$ in the case of deletion.
This allows the adversary to choose a type of action (deletion or insertion) ensuring that 
$\mathbf{Z}(G')\bmod4\in\Set{0,1}$ whatever $\mathbf{Z}(G)\bmod4$ is.
Once the action type is fixed, the adversary chooses a codeword $S^{\downarrow}\in\mathcal{S}^{\downarrow}$
at the Hamming distance 1 from $\mathbf{Y}^{\downarrow}(G)$ and a codeword
$S^{\uparrow}\in\mathcal{S}^{\uparrow}$ at the Hamming distance 1 from $\mathbf{Y}^{\uparrow}(G)$. 
Suppose that $S^{\downarrow}$ and $\mathbf{Y}^{\downarrow}(G)$ differ at coordinate $i$ and that
$S^{\uparrow}$ and $\mathbf{Y}^{\uparrow}(G)$ differ at coordinate $j$.
The adversary changes the adjacency relation between vertices $u$ and $v$ of degree $x$ and $y$ respectively
where $x=d_*+i-1$ and $y=d_*+j-1$ in the case of addition or
$x=d_*+i$ and $y=d_*+j$ in the case of deletion.
By Claim \ref{cl:add-del}, $\mathbf{Y}^{\downarrow}(G')=S^{\downarrow}$ and $\mathbf{Y}^{\uparrow}(G')=S^{\uparrow}$
and, therefore, the modified graph $G'$ has the property $A$.
This proves part 1 of Theorem~\ref{thm:degrees}.

\subsubsection{Proof of Lemma~\ref{lem:pair_of_degrees}}\label{sec412}

We write $I(E)$ to denote the indicator random variable of an event $E$.
Recall that $\xi_1,\ldots,\xi_n$ denote the degrees of vertices $1,\ldots,n$ in $\rg$. 
Due to Theorem~\ref{th:MW}, we can assume that $\xi_1,\ldots,\xi_n$ are independent binomial random variables with 
parameters $n-1$ and $p$. This redefines the random variables $N_y=\sum_{i=1}^nI(\xi_i=y)$
as well as the other random variables defined in terms of $N_y$, in particular, 
$\mathbf{Y}^{\downarrow}$, $\mathbf{Y}^{\uparrow}$, and $\mathbf{Z}$.
Throughout the proof, we use these modified random variables, keeping the same names for them.
Thus, we will prove the statement of Lemma~\ref{lem:pair_of_degrees} for $\mathbf{Y}^{\downarrow}$, $\mathbf{Y}^{\uparrow}$, 
and $\mathbf{Z}$ defined in terms of independent binomial random variables $\xi_1,\ldots,\xi_n$ with 
parameters $n-1$ and $p$. Since this will be done uniformly over all $p$ such that $|p-1/2|\leq\frac{2\ln n}{n}$, 
this will prove the lemma also for the original $\mathbf{Y}^{\downarrow}$, $\mathbf{Y}^{\uparrow}$, 
and $\mathbf{Z}$ on the account of Theorem~\ref{th:MW}
(cf.\ the argument in Section~\ref{sec:th1_proof_finish_probQ_small}).

We define a sequence $\nu=(\nu_1,\ldots,\nu_\delta)$ by
$$
\nu_j=N_{d_*+j-1},
$$
where $N_y$ is as explained above.

\begin{claim}\label{cl:probability_cumulative_divisible}
For each $\seq{r}\in\{0,1,2,3\}^{\delta}$, 
\begin{equation}\label{eq:probability_cumulative_divisible}
\probb{
\nu_j\equiv r_j\Mod4\text{ for all }j=1,\ldots,\delta
}=\frac{1+o(1)}{4^{\delta}},
\end{equation}
where the infinitesimal $o(1)$ is the same for all $\seq{r}$.  
\end{claim}

Let us show that the lemma follows from this claim.
Consider two 0-1-sequences $\vec{\mathbf{B}}=(B_1,\ldots,B_\delta)$ and $\vec{\mathbf{Y}}=(Y_1,\ldots,Y_\delta)$ where
$$
B_j=\nu_j\bmod2\quad\text{ and }\quad Y_j=X_{d_*+j-1}\bmod2,
$$
where $X_y=\sum_{d=d_*}^y N_d$.
Note that $\mathbf{Y}^{\downarrow}$ is the restriction of $\vec{\mathbf{Y}}$ to the first $\delta_1$ coordinates,
and $\mathbf{Y}^{\uparrow}$ is the restriction of $\vec{\mathbf{Y}}$ to the last $\delta_2$ coordinates.
Let $\mathbf{V}^{\downarrow}$ and $\mathbf{V}^{\uparrow}$ denote the analogous restrictions of any sequence $\vec{\mathbf{V}}\in\{0,1\}^\delta$.
The sequences $\vec{\mathbf{B}}$ and $\vec{\mathbf{Y}}$ can be obtained from one another by
$Y_j=B_1\oplus_2\cdots\oplus_2B_j$ and $B_j=Y_j\oplus_2Y_{j-1}$, where $\oplus_2$ is the addition modulo 2. 
Thus, there is a bijection $f$ from $\{0,1\}^\delta$ onto itself such that $\vec{\mathbf{Y}}=f(\vec{\mathbf{B}})$.
We have
$$
\probb{\mathbf{Y}^{\downarrow}\in\mathcal{S}^{\downarrow}}
=\probb{f(\vec{\mathbf{B}})^{\downarrow}\in\mathcal{S}^{\downarrow}}
=\frac{2+o(1)}{\sqrt{n\ln n}}.
$$
The latter equality holds because $\mathcal{S}^{\downarrow}$ is an asymptotically optimal covering code
of length $\delta_1=(\frac12+o(1))\sqrt{n\ln n}$ and because Claim \ref{cl:probability_cumulative_divisible}
implies that 
$$
\probb{B_j=b_j\text{ for all }j=1,\ldots,\delta}=\frac{1+o(1)}{2^{\delta}}
$$
for each fixed sequence $(b_1,\ldots,b_\delta)\in\{0,1\}^\delta$.
The same argument applies to the event $\mathbf{Y}^{\uparrow}\in\mathcal{S}^{\uparrow}$,
yielding part 1 of the lemma.

Moreover, the same argument applies also to the intersection of the two events, yielding
$$
\probb{\mathbf{Y}^{\downarrow}\in\mathcal{S}^{\downarrow},\mathbf{Y}^{\uparrow}\in\mathcal{S}^{\uparrow}}=\frac{4+o(1)}{n\ln n}.
$$ 
Therefore, to obtain part 2 of the lemma, it suffices to prove that
$$
\cprobb{\mathbf{Z}\bmod4=0\text{ or }1}{\mathbf{Y}^{\downarrow},\mathbf{Y}^{\uparrow}}=\frac{1}{2}+o(1).
$$
We prove a formally stronger fact, namely
\begin{equation}
  \label{eq:Y}
\cprobb{\mathbf{Z}\bmod4=0\text{ or }1}{\vec{\mathbf{Y}}}=\frac{1}{2}+o(1).  
\end{equation}
To this end, fix a sequence $(y_1,\ldots,y_\delta)\in\{0,1\}^{\delta}$. 
Let $(b_1,\ldots,b_\delta)=f^{-1}(y_1,\ldots,y_\delta)$.
Consider the set $\mathcal{R}$ of all sequences $(r_1,\ldots,r_\delta)\in\{0,1,2,3\}^{\delta}$ such that
$r_j\bmod2=b_j$ for $j=1,\ldots,\delta$. Furthermore, let $\mathcal{R}^*$ consist of those $(r_1,\ldots,r_\delta)\in\mathcal{R}$
for which 
$$
\sum_{j=1}^{\delta_1+1}(d_*+j-1)r_j\equiv0\text{ or }1\Mod4.
$$
Noting that $\delta_1=\lfloor n/2\rfloor-d_*$ and comparing this sum with the definition \refeq{Z}
of the random variable $\mathbf{Z}$, we conclude by Claim \ref{cl:probability_cumulative_divisible} that
$$
\cprobb{\mathbf{Z}\bmod4=0\text{ or }1}{\vec{\mathbf{Y}}=(y_1,\ldots,y_\delta)}=\frac{|\mathcal{R}^*|}{|\mathcal{R}|}(1+o(1)).
$$
From here we get \refeq{Y} by proving that
\begin{equation}
  \label{eq:RR}
\frac{|\mathcal{R}^*|}{|\mathcal{R}|}=\frac12.  
\end{equation}

For a finite set of integers $S$, let $\uu S$ denote the random variable uniformly distributed on $S$.
We will represent $S$ by listing all its elements. Thus, $\uu 0$ takes on the value 0 with probability 1,
$\uu{0,2}$ takes on each of the values 0 and 2 with probability $1/2$ etc.
Consider the set of four independent random variables
$$
\mathfrak{U}=\Set{\uu0,\uu2,\uu{0,2},\uu{1,3}}.
$$
For two random variables $\uu{}$ and $\uu{}'$ in $\mathfrak{U}$,
consider their sum $\uu{}\oplus_4\uu{}'$ modulo 4.
We allow equal $\uu{}$ and $\uu{}'$, assuming in this case that $\uu{}'$ is an independent copy of $\uu{}$.
A direct inspection reveals that $\uu{}\oplus_4\uu{}'$ always has the same distribution as
one of the random variables in $\mathfrak{U}$. We, therefore, regard $\mathfrak{U}$ as
a commutative semigroup with operation $\oplus_4$. 
Note that $\{0,2\}$ and $\{1,3\}$ are the cosets of the subgroup $\{0,2\}$ in the cyclic group $\mathbb{Z}_4=\{0,1,2,3\}$.
This observation implies the following properties of the semigroup~$\mathfrak{U}$.
\begin{enumerate}[(1)]
\item\label{i:sub} 
$\mathfrak{U}_1=\Set{\uu0,\uu2}$ and $\mathfrak{U}_2=\Set{\uu{0,2},\uu{1,3}}$ are subsemigroups of~$\mathfrak{U}$.
\item\label{i:ideal} 
$\mathfrak{U}_2$ is an ideal in $\mathfrak{U}$, that is, if $\uu{}\in\mathfrak{U}_2$ and $\uu{}'\in\mathfrak{U}$,
then $\uu{}\oplus_4\uu{}'\in\mathfrak{U}_2$.
\end{enumerate}
Moreover, we can multiply the elements of $\mathfrak{U}$ by integer numbers, taking the product modulo~4.
The resulting random variable is also always distributed as a random variable in $\mathfrak{U}$
and, hence, consided an element of $\mathfrak{U}$. Notice the following facts about this scalar multiplication.
\begin{enumerate}[(1)]\setcounter{enumi}{2}
\item \label{i:m0}
$0\uu{}=\uu0$ for every $\uu{}\in\mathfrak{U}$.
\item\label{i:m13}
$1\uu{}=\uu{}$ and $3\uu{}=\uu{}$ for every $\uu{}\in\mathfrak{U}$.
\item\label{i:m2} 
$2\uu{0,2}=\uu0$ and $2\uu{1,3}=\uu2$.
\end{enumerate}

Now, consider independent random variables $\mathbf{r}_1,\ldots,\mathbf{r}_\delta$
such that $\mathbf{r}_j$ is distributed as $\uu{0,2}$ if $b_j=0$ and as $\uu{1,3}$ if $b_j=1$. 
The random sequence $(\mathbf{r}_1,\ldots,\mathbf{r}_\delta)$ is uniformly distributed on $\mathcal{R}$.
It follows that
\begin{equation}
  \label{eq:RRP}
\frac{|\mathcal{R}^*|}{|\mathcal{R}|}=
\probb{\sum_{j=1}^{\delta_1+1}(d_*+j-1)\mathbf{r}_j=0\text{ or }1},
\end{equation}
where the sum is a random variable in $\mathfrak{U}$ computed according to the arithmetic rules introduced above.
The sum of all terms with scalars equal to 1 or 3 modulo 4 evaluates to an element in $\mathfrak{U}_2$
by Properties \ref{i:m13} and \ref{i:sub}. The sum of all terms with scalars equal to 0 or 2 modulo 4 evaluates 
to an element in $\mathfrak{U}_1$ by Properties \ref{i:m0}, \ref{i:m2}, and \ref{i:sub}. 
Finally, Property \ref{i:ideal} implies that the whole sum evaluates to an element in $\mathfrak{U}_2$
and, therefore, the probability in \refeq{RRP} is equal to $1/2$.
Equality \refeq{RR} is therewith proved.


To complete the proof of the lemma, it remains to prove Claim \ref{cl:probability_cumulative_divisible}.

\begin{proof}[Proof of Claim \ref{cl:probability_cumulative_divisible}] 
For the characteristic function of the $\delta$-dimensional random vector $\nu-\seq{r}$, we have
\begin{align*}
 \phi_{\nu-\seq{r}}(\lambda)=
 {\sf E}\exp\left[i\sum_{j=1}^{\delta}(\nu_j-r_j)\lambda_j\right]
 & =
 \exp\left[-i\sum_{j=1}^{\delta}r_j\lambda_j\right]
 {\sf E}\exp\left[i\sum_{v=1}^n\sum_{j=1}^{\delta}I(\xi_v=d_*+j-1)\lambda_j\right]\\
  & = 
  \exp\left[-i\sum_{j=1}^{\delta}r_j\lambda_j\right] \left({\sf E}\exp\left[i\sum_{j=1}^{\delta}I(\xi_1=d_*+j-1)\frac{\lambda_j\pi}{2}\right]\right)^n.
\end{align*}
Let $\vec{0}$ be the vector consisting of $\delta$ zeros. The vectors $\vec{1}$, $\vec{2}$, and $\vec{3}$ are introduced similarly. 
The probability on the left hand side of \refeq{probability_cumulative_divisible} is equal to 
$4^{-\delta}\big(1+\sum_{\lambda\in\{0,1,2,3\}^{\delta}\setminus\{\vec{0}\}}\phi_{\nu-\seq{r}}\left(\frac{\lambda\pi}{2}\right)\big)$ 
by part 2 of Lemma~\ref{lem:character}.

Let us first estimate $\phi_{\nu-\seq{r}}\left(\frac{\lambda\pi}{2}\right)$ for $\lambda\in\{\vec{1},\vec{2},\vec{3}\}$. 
Let $R=D_n$. By the de Moivre--Laplace limit theorem,
$$
 \prob{\xi_1\notin R}=\frac{(1+o(1))e^{\sqrt{\ln n}}}{\sqrt{2\pi n\ln n}}.
$$
If $\lambda=\vec{1}$, then
\begin{align*}
\left({\sf E}\exp\left[i\sum_{j=1}^{\delta}I(\xi_1=d_*+j-1)\frac{\lambda_j\pi}{2}\right]\right)^n & =
\left({\sf E}\exp\left[i I(\xi_1\in R)\frac{\pi}{2}\right]\right)^n \\
& =\biggl(\prob{\xi_1\notin R}+i(1-\prob{\xi_1\notin R})\biggr)^n.
\end{align*}
Therefore, in this case we have
$$
\left|\phi_{\nu-\seq{r}}\left(\frac{\lambda\pi}{2}\right)\right|=
\left(1-2\prob{\xi_1\notin R}+2\left(\prob{\xi_1\notin R}\right)^2\right)^{n/2}=
\exp\left[-(1+o(1))\frac{\sqrt{n}e^{\sqrt{\ln n}}}{\sqrt{2\pi\ln n}}\right]=o(1).
$$
Furthermore, 
$\left|\phi_{\nu-\seq{r}}\left(\frac{\vec{3}\pi}{2}\right)\right|=\left|\phi_{\nu-\seq{r}}\left(\frac{\vec{1}\pi}{2}\right)\right|=o(1)$.
Finally, if $\lambda=\vec{2}$, then
$$
\left({\sf E}\exp\left[i\sum_{j=1}^{\delta}I(\xi_1=d_*+j-1)\frac{\lambda_j\pi}{2}\right]\right)^n=
\left({\sf E}\exp\left[i I(\xi_1\in R)\pi\right]\right)^n=\left(2\prob{\xi_1\notin R}-1\right)^n.
$$
and 
$$
\left|\phi_{\nu-\seq{r}}\left(\frac{\lambda\pi}{2}\right)\right|=\left(1-2\prob{\xi_1\notin R}\right)^{n/2}=o(1)
$$ 
like in the preceding two cases.

Now, consider $\lambda\in\{0,1,2,3\}^{\delta}\setminus\{\vec{0},\vec{1},\vec{2},\vec{3}\}$. 
For $r=1,2,3$, we define the sets $R_r=R_r(\lambda)$ by $R_r=\Set{d_*+j-1}{\lambda_j=r}$.
Also, $R_0=\{0,\ldots,n-1\}\setminus(R_1\cup R_2\cup R_3)$. Note that
$$ 
\left({\sf E}\exp\left[i\sum_{j=1}^{\delta}I(\xi_1=d_*+j-1)\frac{\lambda_j\pi}{2}\right]\right)^n=
 \left({\sf E}\exp\left[i\sum_{r=0}^{3}I(\xi_1\in R_r)\frac{r\pi}{2}\right]\right)^n.
$$
For non-negative reals $p_0,p_1,p_2,p_3$ such that $p_0+p_1+p_2+p_3=1$ and arbitrary reals 
$\varphi_0,\varphi_1,\varphi_2,\varphi_3\in\mathbb{R}$ we have
$$
 \left(\sum_{r=0}^3p_r\cos\varphi_r\right)^2+
 \left(\sum_{r=0}^3p_r\sin\varphi_r\right)^2=
 1-2\sum_{r_1\neq r_2}p_{r_1}p_{r_2}\sin^2\frac{\varphi_{r_1}-\varphi_{r_2}}{2},
$$
which implies that
\begin{align*}
 \left|{\sf E}\exp\left[i\sum_{r=0}^{3}I(\xi_1\in R_r)\frac{r\pi}{2}\right]\right|^2 
 &=
 1-2\sum_{r_1\neq r_2}\prob{\xi_1\in R_{r_1}}\prob{\xi_1\in R_{r_2}}\sin^2\frac{\pi(r_1-r_2)}{4}\\
 & \leq
 1-\sum_{r_1\neq r_2}\prob{\xi_1\in R_{r_1}}\prob{\xi_1\in R_{r_2}}\\
 &=
 1-\sum_{r=0}^3\prob{\xi_1\in R_r}(1-\prob{\xi_1\in R_r})=
 \sum_{r=0}^3(\prob{\xi_1\in R_r})^2.
\end{align*}

Let $\Lambda$ be the set of all $\lambda\in\{0,1,2,3\}^{\delta}\setminus\{\vec{0},\vec{1},\vec{2},\vec{3}\}$ with the property that
$\prob{\xi_1\in R_{r_0}}\in\left[\frac{\ln n}{\sqrt{n}},1-\frac{\ln n}{\sqrt{n}}\right]$ for some $r_0\in\{0,1,2,3\}$.
If $\lambda\in\Lambda$, then
$$
 \sum_{r=0}^3(\prob{\xi_1\in R_r})^2\leq(\prob{\xi_1\in R_{r_0}})^2+(1-\prob{\xi_1\in R_{r_0}})^2\leq 1-\frac{\ln n}{\sqrt{n}}
$$
for $n$ large enough. Therefore,
\begin{align}
 \left|\sum_{\lambda\in\Lambda}\phi_{\nu-\seq{r}}(\lambda)\right|
 &\leq\sum_{\lambda\in\Lambda}\left|{\sf E}\exp\left[i\sum_{r=0}^{3}I(\xi_1\in R_r)\frac{r\pi}{2}\right]\right|^n \notag \\
 &\leq\sum_{\lambda\in\Lambda}\left(1-\frac{\ln n}{\sqrt{n}}\right)^{n/2}\notag\\ 
 &\leq 4^{\delta}\exp\left[-\frac{1}{2}\sqrt{n}\ln n\right]=\exp\left[-\frac{1+o(1)}{2}\sqrt{n}\ln n\right].
 \label{character_complicated}
\end{align}

Without loss of generality, suppose that $\prob{\xi_1\in R_0}\leq\prob{\xi_1\in R_1}\leq\prob{\xi_1\in R_2}\leq\prob{\xi_1\in R_3}$.
Let $\lambda\notin\Lambda\cup\{\vec{0},\vec{1},\vec{2},\vec{3}\}$. Assuming that $n$ is large enough, we get
$\prob{\xi_1\in R_2}<\frac{\ln n}{\sqrt{n}}$, $\prob{\xi_1\in R_3}> 1-\frac{\ln n}{\sqrt{n}}$. Therefore, 
\begin{align*}
 \sum_{r=0}^3(\prob{\xi_1\in R_r})^2 & \leq
 (\prob{\xi_1\in R_3})^2+(1-\prob{\xi_1\in R_3})^2\\
 &=1-2\prob{\xi_1\in R_3}+2(\prob{\xi_1\in R_3})^2\\
 &\leq\prob{\xi_1\in R_3}=1-(\prob{\xi_1\in R_0}+\prob{\xi_1\in R_1}+\prob{\xi_1\in R_2}).
\end{align*} 
For $r=0,1,2,3$, let $n_r$ be the number of $j$ such that $\lambda_j=r$. The de Moivre--Laplace limit theorem implies that 
$$
\prob{\xi_1\in R_r}\geq\sum_{d\in R_r\cap D_n}\prob{\xi_1=d}\geq n_r\prob{\xi_1=d_*}(1+o(1))\geq\frac{n_r \exp\left[\sqrt{\ln n}\right]}{n\sqrt{\pi/2}}(1+o(1)).
$$
Finally, summing up over all $\lambda\in\{0,1,2,3\}^{\delta}\setminus(\Lambda\cup\{\vec{0},\vec{1},\vec{2},\vec{3}\})$,
we obtain
\begin{multline*}
 \left|\sum_{\lambda}\phi_{\nu-\seq{r}}(\lambda)\right|
 \leq\sum_{\lambda}\left|{\sf E}\exp\left[i\sum_{r=0}^{3}I(\xi_1\in R_r)\frac{r\pi}{2}\right]\right|^n \\
 \leq\sum_{n_0,n_1,n_2}{\delta\choose n_0,n_1,n_2,n_3}\left(1-\frac{(n_0+n_1+n_2) \exp\left[\sqrt{\ln n}\right]}{n\sqrt{\pi/2}}(1+o(1))\right)^{n/2}\\
 \leq \sum_{n_0,n_1,n_2}
 \exp\left[(n_0+n_1+n_2)\left(\ln\delta-\exp\left[\sqrt{\ln n}\right]\left(\sqrt{\frac{1}{2\pi}}+o(1)\right)\right)\right]\\
 =\exp\left[-\exp\left[\sqrt{\ln n}\right]\left(\sqrt{\frac{1}{2\pi}}+o(1)\right)\right].
\end{multline*}
Along with (\ref{character_complicated}), this estimate completes the proof of Claim~\ref{cl:probability_cumulative_divisible}.
\end{proof}

\subsubsection{On the possibility of a sharp bound}
It is actually possible to get an optimal upper bound $(1+o(1))/(n\ln n)$ for the probability in Theorem~\ref{thm:degrees} 
if we assume that there exists an asymptotically optimal covering code of length $m$ and radius $2$ for a sufficiently 
dense sequence of values of $m$. Recall that a covering code of radius $2$ is a subset of $\{0,1\}^m$ such that every word 
in $\{0,1\}^m$ is at Hamming distance at most $2$ from some word in the code. It is asymptotically optimal if the sum of 
sizes of Hamming balls of radius $2$ around the codewords is $(1+o(1))2^m.$ Since each Hamming ball of radius 2 contains 
$(1+o(1))m^2/2$ words, such a code contains $(1+o(1)) 2^{m+1}/m^2$ codewords. For every radius $\rho\geq 2$, the conjecture a
bout the existence of an asymptotically optimal covering code of radius $\rho$ is still open (the best known upper bounds 
for densities of covering codes can be found in~\cite{KrivelevichSudakov}).

Let us sketch an upper bound  argument, borrowing the notation of Section~\ref{sec41}.  
Consider the vector $\mathbf Y= (X_{d_*},\ldots, X_{d^*-1})\bmod2$ and assume that $K\subset \{0,1\}^m$, for $m=d^*-d_*$,
is a code as in the preceding paragraph. Also, let $\mathbf Z=\sum_{d_*\leq y\leq d^*} yN_y$.
Then the desired property $\mathcal P$ is defined by the conditions that $\mathbf Y\in K$ and $\mathbf Z\bmod4\in\{0,1\}$.
Similarly to Section~\ref{sec412},  it is possible to show that 
the property $\mathcal P$ is satisfied with probability $(1+o(1))|K|/2^{m+1} = (1+o(1))/m^2$, which is exactly the bound that we need. 
It is also easy to check that, for almost all graphs, $\mathcal P$ can be satisfied by changing only one edge.

Finally, note that for our purposes it is actually sufficient to have an asymptotically optimal {\it almost}-covering code of 
radius $2$, i.e., a set of codewords in $\{0,1\}^m$ whose Hamming balls of radius $2$ cover almost all~$\{0,1\}^m$.

It is an intriguing question whether having an asymptotically optimal almost-covering code of radius $2$ is actually equivalent to having such a random graph property.

\subsection{Lower bound}\label{ss:deg-seq-lower}

Using the coding-theoretic terminology, we say that a graph $G$ \emph{covers} a graph $G'$
if $G'$ is obtained from $G$ by changing the adjacency relation between two vertices.
All graphs considered in this section are on the vertex set $[n]=\{1,\ldots,n\}$.
Let $\Gamma(G)$ denote the set of all $n\choose2$ graphs covered by $G$.
For a set of graphs $\cQ$, let $\Gamma(\cQ)=\bigcup_{G\in\cQ}\Gamma(G)$.
Our proof strategy is based on the following simple observation.
If $\cQ$ is an anti-stochastic property, then $|\Gamma(\cQ)|=(1-o(1))2^{n\choose2}$
and, therefore, if we also have 
\begin{equation}
  \label{eq:GammaQ}
|\Gamma(\cQ)|\le(1+o(1))n\ln n|\cQ|,   
\end{equation}
this readily implies the desired lower bound $\prob{\rg\in\cQ}\ge\frac{1+o(1)}{n\ln n}$.
Moreover, for a set of graphs $\cA$, let $\Gamma_\cA(\cQ)=\Gamma(\cQ)\cap\cA$.
If $\rg\in\cA$ whp, i.e., $|\cA|=(1-o(1))2^{n\choose2}$, then istead of \refeq{GammaQ}
it suffices to have the formally weaker inequality
\begin{equation}
  \label{eq:GammaAQ}
|\Gamma_\cA(\cQ)|\le(1+o(1))n\ln n|\cQ|.   
\end{equation}
We prove that \refeq{GammaAQ} holds for every anti-stochastic property $\cQ$
expressible in terms of vertex degrees when $\cA$ is chosen to be the set
of graphs satisfying Conditions 1--4 stated in Lemma~\ref{lem:degrees_range}.

Let $\lds(G)$ denote the labeled degree sequence of a graph $G$.
For two sequences of integers $\seq{d}$ and $\seq{s}$, we write $\seq{s}\simeq\seq{d}$
if $\seq{s}$ is a permutation of $\seq{d}$.
For a sequence of integers $\seq{d}=(d_1,\ldots,d_n)$, let $\cG(\seq{d})$ denote
the set of all graphs $G$ with $\lds(G)\simeq\seq{d}$. Since $\cQ$ is expressible 
in terms of vertex degrees, $\cQ$ is the disjoint union of the sets $\cG(\seq{d})$ over
all $\seq{d}=\lds(G)$ for $G\in\cQ$. Therefore, it is enough to prove that
\begin{equation}
  \label{eq:GammaAGd}
|\Gamma_\cA(\cG(\seq d))|\le(1+o(1))n\ln n|\cG(\seq d)|  
\end{equation}
for every such $\seq{d}$, and this is what we do in the rest of the proof.

Let $p(\seq d)$ denote the number of sequences $\seq{s}$ such that $\seq{s}\simeq\seq{d}$.
By $n_y$ we denote the frequency of the degree $y$ in $\seq d$.
Note that $p(\seq d)=n!/\prod_y n_y!$.

According to the notation in Subsection \ref{ss:degrees}, $g(\seq d)$ is the number of all graphs with $\lds(G)=\seq{d}$.
Note that $g(\seq d)=g(\seq s)$ if $\seq{s}\simeq\seq{d}$.
Therefore, 
\begin{equation}
  \label{eq:pg}
|\cG(\seq d)|=p(\seq d)g(\seq d).  
\end{equation}

We use the shortcuts
$$
a=\sqrt{{\textstyle\frac12}n\ln n},\quad b=\frac{1}{2}\sqrt{n\ln n\left(1-\zeta_1\right)},
\quad\text{and}\quad c=\frac{1}{2}\sqrt{n\ln n\left(1-\zeta_2\right)},
$$
where $\zeta_1$ and $\zeta_2$ are the constants in Lemma~\ref{lem:degrees_range}.
Obviously, $a>b>c$.

If $\seq d\ne\lds(G)$ for any $G\in\Gamma(\cA)$, then $\Gamma_\cA(\cG(\seq d))=\emptyset$, and \refeq{GammaAGd}
is trivially true. In other words, the graphs not covered by $\cA$ themselves do not cover
any graph in $\cA$, and their degree sequences do not need to be considered.
We, therefore, suppose that 
\begin{equation}
  \label{eq:dG}
\seq d=\lds(G)\text{ for some }G\in\Gamma(\cA).  
\end{equation}
This imposes the following conditions on $\seq d$.
\begin{enumerate}
\item[(P1)] 
$|d_i-n/2|<a+1$ for every vertex $i\in[n]$.
\item[(P2)]
$|d_i-n/2|\le b$ for all but at most $4\sqrt{n}\ln^{1/4}n+2$ vertices $i$.
\item[(P3)]
     \begin{enumerate}[(i)]
      \item 
If $|y-n/2| < c+1$, then
    $n_y\ge \ln^{5/4}n-2$ and 
    $|n_y-n_{y+1}|\leq\frac{n_y}{\ln\ln n}+5$.
      \item 
    If $c\leq|y-n/2|\le b$, then $n_y\le \ln^2n+2$.
\end{enumerate}
\item[(P4)]
For all but at most $\sqrt{n}+2$ vertices $i$, the condition
    $c\le|d_i-n/2|\le b$
    implies that
    $$
    n_y\ge\ln^{3/4}n-2\,\,\text{ as well as }\,\,
    |n_y-n_{y+1}|\le\frac{n_y}{\ln\ln n}+5\,\,\text{ and }\,\,
    |n_y-n_{y-1}|\le\frac{n_y}{\ln\ln n}+5
    $$
for every of three values $y\in\{d_i-1,d_i,d_i+1\}$.
\end{enumerate}
For each $k=1,2,3,4$, Condition (P$k$) follows from Part $k$ of Lemma~\ref{lem:degrees_range}
by noting that switching the adjacency of a vertex pair $\{i,j\}$ in $G$ 
affects only two degrees $d_i$ and $d_j$ (which either increase or decrease by 1) and 
at most four counts $n_{d_i}$, $n_{d_i+1}$, $n_{d_j}$, and $n_{d_j+1}$ if the edge $\{i,j\}$
is added or $n_{d_i}$, $n_{d_i-1}$, $n_{d_j}$, and $n_{d_j-1}$ if this  edge is deleted
(each of the counts can be changed by at most 2).

\begin{claim}\label{cl:g}
Suppose that $G\in\Gamma(\cA)$ and $G'\in\Gamma(G)$.
Let $\seq d=\lds(G)$ and $\seq d'=\lds(G')$.
Then $g(\seq d')=(1+o(1))g(\seq d)$.
\end{claim}

\begin{proof}
The claim is a direct consequence of Property (P1) and Theorem \ref{th:counting_graps_deg_sequence}.
Indeed, let $f(\mu)=\ln\of{\mu^{\mu}(1-\mu)^{1-\mu}}$ and note that $f'(\mu)=\ln(\mu/(1-\mu))$. 
Let $\gamma$ and $\mu$ be the parameters of $\seq d$ as in Theorem \ref{th:counting_graps_deg_sequence}
and $\gamma'$ and $\mu'$ be their values for $\seq d'$. Property (P1) implies that $\gamma=O(\ln n)$, $\mu=1/2+o(1)$,
$|\gamma-\gamma'|=O(\ln^{1/2}n/n^{3/2})$, and $|\mu-\mu'|=O(1/n^2)$. By Theorem~\ref{th:counting_graps_deg_sequence}
and Property (P1), we obtain
\begin{multline*}
\frac{g(\seq d')}{g(\seq d)}=(1+o(1))\exp\ofsq{O((\gamma+\gamma')(\gamma-\gamma'))+(f(\mu')-f(\mu))n(n-1)/2}\\
=(1+o(1))\exp\ofsq{O(\ln^{3/2}n/n^{3/2})+O(f'(\mu)(\mu'-\mu))n^2}\\
=(1+o(1))\exp\ofsq{o(1)+O(\ln(\mu/(1-\mu)))}=(1+o(1))\exp\ofsq{o(1)}=1+o(1),
\end{multline*}
as claimed.
\end{proof}

\begin{definition}\label{def:good}
  Let $\seq{d}=(d_1,\ldots,d_n)$ be as in \refeq{dG} and $d=d_i$ for some $i$.
We call the degree~$d$ \emph{good} if $|d-n/2|\le b$ and
the following conditions are true for $y=d$:
\begin{enumerate}[(1)]
\item\label{item:1}
$n_y\ge\ln^{3/4}n-2$,
\item\label{item:2}
$|n_y-n_{y+1}|\le\frac{n_y}{\ln\ln n}+5$, and
\item\label{item:3} 
at least one of the following two conditions holds true:
\begin{eqnarray*}
    |n_y-n_{y-1}|&\le&\frac{n_y}{\ln\ln n}+5\text{ or}\\
    \text{both }|n_y-n_{y-1}|&\le&\frac{n_{y-1}}{\ln\ln n}+5\text{ and }n_{y-1}\ge\ln^{3/4}n-2.
\end{eqnarray*}
\end{enumerate}

Furthermore, a vertex $i\in[n]$ is called \emph{good} if its degree $d_i$ is good.
Finally, a vertex $i$ (resp.\ degree $d$) is \emph{bad} if $i$ (resp.~$d$) is not good.
\end{definition}

Definition \ref{def:good} is much motivated by the following fact.

\begin{claim}\label{cl:p}
Let $\seq d=\lds(G)$, $\seq d'=\lds(G')$, and $G'\in\Gamma(G)$.
Specifically, suppose that $G'$ is obtained by switching the adjacency of a vertex pair $\{i,j\}$ in $G$.
If the vertices $i$ and $j$ are good, then $p(\seq d')=(1+o(1))p(\seq d)$.  
\end{claim}

\begin{proof}
  We use the equality $p(\seq d)=n!/\prod_zn_z!$.
Let $x=d_i$ and $y=d_j$ be the degrees of $i$ and $j$ in $G$ and suppose that $x\ne y$. 
Assume that $i$ and $j$ are not adjacent in $G$, that is, $G'$ is obtained by adding the edge $\{i,j\}$.
We have
$$
\frac{p(\seq d)}{p(\seq d')}=\frac{(n_x-1)!(n_{x+1}+1)!(n_y-1)!(n_{y+1}+1)!}{n_x!n_{x+1}!n_y!n_{y+1}!}
=\frac{n_{x+1}+1}{n_x}\frac{n_{y+1}+1}{y_x}=1+o(1),
$$
where the last equality follows from Conditions (\ref{item:1}) and (\ref{item:2}) in Definition \ref{def:good}.
If $G'$ is obtained by removing the edge $\{i,j\}$, then
$$
\frac{p(\seq d)}{p(\seq d')}=\frac{n_{x-1}+1}{n_x}\frac{n_{y-1}+1}{y_x}=1+o(1),
$$
where we use Conditions (\ref{item:1}) and (\ref{item:3}). The case of $x=y$ is similar. 
\end{proof}

The following definition is a strengthening of Definition~\ref{def:good}.

\begin{definition}\label{def:very-good}
 We call a degree $d$ \emph{very good} if $|d-n/2|\le b-1$ and
Conditions (\ref{item:1})--(\ref{item:3}) in Definition \ref{def:good} are true not only for $y=d$
but also for $y=d-1$ and $y=d+1$.
Furthermore, a vertex $i$ is called \emph{very good} if its degree $d_i$ is very good.
\end{definition}

We make a simple observation for further references.

\begin{claim}\label{cl:very-good}
  If a degree $d$ is very good, then the degrees $d-1$, $d$, and $d+1$ are good.
\end{claim}

We also state a fact that will play an important role below.

\begin{claim}\label{cl:many-very-good}
If $G\in\Gamma(\cA)$, then all but $O(\sqrt n\ln^{1/4}n)$ vertices of $G$ are very good.
\end{claim}

\begin{proof}
Let $\seq d=\lds(G)$. According to Property (P2) and part (ii) of Property (P3),
the inequality $|d_i-n/2| > b-1$ can be true for at most $4\sqrt{n}\ln^{1/4}n+\ln^2n+4$ vertices $i$.
Consider the vertices $i$ with $|d_i-n/2| \le b-1$. If $n/2-c+1 < d_i < n/2+c$, then
Conditions (\ref{item:1})--(\ref{item:3}) of Definition \ref{def:good}
are true for all three values $y\in\{d_i-1,d_i,d_i+1\}$ by part (i) of Property (P3). 
Part (i) of Property (P3) ensures that $n/2-c < d_i \le n/2-c+1$ for at most $\ln^2n+2$ vertices $i$.
Thus, it remains to consider the vertices $i$ with $c \le |d_i-n/2| \le b-1$.
Property (P4) guarantees Conditions (\ref{item:1})--(\ref{item:3}) for the three values of $y$
for all but at most $\sqrt{n}+2$ of such vertices.
\end{proof}

Recall that our goal is to prove that the bound \refeq{GammaAGd} is true under the condition \refeq{dG}.
In fact, we will prove a stronger bound
\begin{equation}
  \label{eq:GammaGd}
|\Gamma(\cG(\seq d))|\le(1+o(1))n\ln n|\cG(\seq d)|  
\end{equation}
for the degree sequence $\seq d$ of a graph in $\Gamma(\cA)$.  
We split the set of the graphs covered by a graph $G$ in three parts, namely
$$
\Gamma(G)=\Gamma_+(G)\cup\Gamma_\pm(G)\cup\Gamma_-(G).
$$
Specifically, $\Gamma_+(G)$ consists of the graphs obtained from $G$ by switching the
adjacency relation between two good vertices $i$ and $j$.
The graphs in $\Gamma_\pm(G)$ are obtained by switching a pair $\{i,j\}$
such that $i$ is bad and $j$ is very good. Finally, the graphs in $\Gamma_-(G)$ are obtained
from $G$ by switching $\{i,j\}$ such that $i$ is bad and $j$ is not very good.
Extending this notation to sets of graphs, we have
\begin{equation}
  \label{eq:Gamma-split}
\Gamma(\cG(\seq d))=\Gamma_+(\cG(\seq d))\cup\Gamma_\pm(\cG(\seq d))\cup\Gamma_-(\cG(\seq d)).
\end{equation}
To prove \refeq{GammaGd}, we treat each of the three parts separately.

\begin{claim}\label{cl:+}
$|\Gamma_+(\cG(\seq d))|\le(1+o(1))n\ln n|\cG(\seq d)|$.
\end{claim}

\begin{proof}
  Let $U$ be an inclusion-maximal set of sequences $\seq s$ such that $\seq s=\lds(G)$
for some $G\in\Gamma_+(\cG(\seq d))$ and such that $\seq s\simeq\seq s'$ for no two
sequences $\seq s$ and $\seq s'$ in $U$. According to Definition \ref{def:good},
$\seq d$ contains no more than $2b+1$ good degrees. It follows that
$$
|U|\le2(2b+1)+2{2\lfloor b\rfloor+1\choose2}=(1+o(1))n\ln n.
$$
Here, the former term corresponds to switching a vertex pair $\{i,j\}$ in some $G$ with $\lds(G)=\seq d$
such that $d_i=d_j$, while the latter term corresponds to the case $d_i\ne d_j$.
The factor of 2 in each term corresponds to the two cases of adding and deleting $\{i,j\}$.

If $\seq d'\in U$, then $g(\seq d')=(1+o(1))g(\seq d)$ by Claim \ref{cl:g}
and $p(\seq d')=(1+o(1))p(\seq d)$ by Claim \ref{cl:p}. Therefore,
$|\cG(\seq d')|=(1+o(1))|\cG(\seq d)|$ by Equality \refeq{pg}. We conclude that
$$
|\Gamma_+(\cG(\seq d))|=|\bigcup_{\seq d'\in U}\cG(\seq d')|\le|U|\max_{\seq d'\in U}|\cG(\seq d')|
=(1+o(1))n\ln n|\cG(\seq d)|,
$$
as claimed.
\end{proof}

\begin{claim}\label{cl:-}
$|\Gamma_-(\cG(\seq d))|=o(n\ln n)|\cG(\seq d)|$.
\end{claim}

\begin{proof}
For $G\in\cG(\seq d)$, the number of graphs in $\Gamma_-(G)$ is bounded by the number of bad vertices
multiplied by the number of not very good vertices in $G$. 
It follows by Claim \ref{cl:many-very-good} that $|\Gamma_-(G)|=O(n\ln^{1/2}n)$.
\end{proof}

\begin{claim}\label{cl:pm}
$|\Gamma_\pm(\cG(\seq d))|=o(n\ln n)|\cG(\seq d)|$.
\end{claim}

\begin{proof}
  Let $\seq s\simeq\seq d$ and $\seq s'\simeq\seq d$. We write $\seq s\equiv_{\seq d}\seq s'$
if $s_i=s_i'$ for every bad vertex $i$. Let $C$ be an $\equiv_{\seq d}$-equivalence class
of sequences. We write $\bv(C)$ to denote the set of all $i$ such that $i$
(or, equivalently $s_i$) is bad in any sequence $\seq s\in C$.
Furthermore, $\cG_C(\seq d)$ is the set of all graphs $G$ such that $\lds(G)\in C$.
Since $\cG(\seq d)=\bigcup_C\cG_C(\seq d)$, where the union is over all $\equiv_{\seq d}$-equivalence 
classes $C$, it suffices to show that
\begin{equation}
  \label{eq:Gamma-pm}
|\Gamma_\pm(\cG_C(\seq d))|=o(n\ln n)|\cG_C(\seq d)|
\end{equation}
for every $C$. 

For $i\in\bv(C)$, we define $\Gamma_\pm^i(\cG_C(\seq d))$ to be
the set of all graphs obtained from a graph in $\cG_C(\seq d)$ by switching
a vertex pair $\{i,j\}$ for some very good $j$.
The labeled degree sequence of a graph in $\Gamma_\pm^i(\cG_C(\seq d))$ can be obtained
from $\seq d$ in the following way.
\begin{enumerate}
\item 
The good degrees of $\seq d$ are permuted, while the bad degrees stay fixed.
That is, only positions of $\seq d$ outside $\bv(C)$ are permuted. Let $\seq s$ be the resulting sequence.
This step corresponds to choosing a graph $G\in\cG_C(\seq d)$ with $\lds(G)=\seq s$.
\item 
Choose $j\notin\bv(C)$ such that $s_j$ is very good. The degrees $s_i$ and $s_j$ are both
either increased or decreased by 1, which corresponds to the cases of adding or deleting
the edge $\{i,j\}$ 
\end{enumerate}
Note that the order of Steps 1 nad 2 can be changed. More specifically, 
exactly the same set of sequences is obtained in this way:
\begin{enumerate}
\item[$1'$.]
A very good degree $d$ is chosen. Let $j\in[n]$ be the least integer such that $d_j=d$.
The degrees $d_i$ and $d_j$ are both either increased or decreased by~1.
\item[$2'$.]
The positions of the resulting sequence outside $\bv(C)$ are permuted.
\end{enumerate}
According to Definition \ref{def:very-good}, Step $1'$ can be done in at most
$2(2b-1)=O(\sqrt{n\ln n})$ ways. Suppose that the outcome of Step $1'$ is fixed.
By Claim \ref{cl:very-good}, both degrees $d-1$ and $d+1$ are good.
It follows that the number of sequences produced in Step $2'$ is equal to
$$
\frac{(n-|\bv(C)|)!}{
(n_d-1)!(n_{d+1}+1)!\prod\limits_{y\ne d,d+1}n_y!
}
$$
if $d=d_j$ is increased by 1 or to
$$
\frac{(n-|\bv(C)|)!}{
(n_d-1)!(n_{d-1}+1)!\prod\limits_{y\ne d,d-1}n_y!
}
$$
if $d$ is decreased by 1, where the products are over good degrees.
Once again using the fact that the degrees $d-1$ and $d+1$ are good,
quite similarly to Claim \ref{cl:p} we see that this number is
$$
(1+o(1))\frac{(n-|\bv(C)|)!}{
\prod_y n_y!
}
=(1+o(1))|C|,
$$
where the product in the left hand side is over all good degrees.
Taking into account Claim \ref{cl:g}, we conclude that 
$$
|\Gamma_\pm^i(\cG_C(\seq d))|=O(\sqrt{n\ln n})|C|g(\seq d)=O(\sqrt{n\ln n})|\cG_C(\seq d)|
$$
and, therefore, 
$$
|\Gamma_\pm(\cG_C(\seq d))|=O(|\bv(C)|\sqrt{n\ln n})|\cG_C(\seq d)|=O(n\ln^{3/4} n)|\cG_C(\seq d)|,
$$
where the last equality is due to Claim \ref{cl:many-very-good}.
Equality \refeq{Gamma-pm} follows.
\end{proof}
  
Bound \refeq{GammaGd} now follows readily from Claims \ref{cl:+}, \ref{cl:-}, and \ref{cl:pm}
by Equality \refeq{Gamma-split}. The proof of part 2 of Theorem \ref{thm:degrees} is complete.

\section{Conclusion and further questions}

Anti-stochastic properties of graphs studied in this paper are a natural concept in the
context of the research on network vulnerability \cite{Schaeffer21}.
Our focus on isomorphism-invariant properties (or, equivalently, on \emph{unlabeled} graphs)
is motivated by an observation that, in realistic scenarios, an adversary can only
be interested in forcing somehow meaningful, structured properties.
Theorem \ref{thm:unlab} determines the optimum probability of an anti-stochastic property in this setting,
and Theorem \ref{thm:degrees} concerns the even more contstrained scenario when an adversary aims at
an anti-stochastic property expressible solely in terms of vertex degrees.

There are several further questions naturally arising in this context.
We here consider the random graph model $G(n,p)$
in the case $p=1/2$. Not all tools in our analysis
can be directly applied to other edge probabilities $p=p(n)$. For example, the assumption $p=1/2$
is essentially used in our proof of the lower bound in Theorem \ref{thm:degrees}.
Other random graph models, especially those designed to describe real-life networks 
(e.g., \cite{BollobasR03,BollobasR03a,FlaxmanFV07}), would also be of considerable interest.

We resrict ourselves to the case of a limited adversary who is able to change the
adjacency relation just between a single pair of vertices. Consideration of other
perturbation types, like changing adjacencies of multiple pairs, vertex/edge deletions or insertions, etc.,
would be also well motivated (corresponding to various types of errors
studied in coding theory~\cite{Firer21}).

Another question concerns the descriptive complexity of anti-stochastic properties. 
As we mentioned in Remark \ref{rem:descr}, the
asymptotically optimal anti-stochastic property constructed in the proof of Theorem \ref{thm:unlab}
is expressible in the 2-variable infinitary logic $C^2_{\infty\omega}$.
On the other hand, it can be shown that no anti-stochastic property
is expressible in first-order logic \cite{Akhmejanova}. 
Curiously, in \cite{Akhmejanova} it is also shown that, for some sequence $(n_i)$, there exists an 
anti-stochastic property of $\mathbf{G}_{n_i}$ expressible in existential monadic second order logic.
Other positive and negative results
in this direction could be interesting as they would be a way to formally show
what kind of properties can be and what cannot be achieved by an adversary.

\subsection*{Acknowledgments}
We would like to thank Grigorii Kabatiansky for useful discussions of the construction of asymptotically optimal covering codes in~\cite{KabatyanskiiP88}. 

Oleg Verbitsky is supported by DFG grant KO 1053/8--2. Part of this work was done while Andrey Kupavskii and Maksim Zhukovskii were at St Petersburg State University, supported by Russian Science Foundation, project 22-11-00131.



\bibliographystyle{abbrv}
\bibliography{main.bib}

\appendix

\section*{Appendix}

\section{Proof of Lemma~\ref{lem:degrees_range}}

\hspace{17pt}1. The proof can be found in~\cite[Corollary 7]{BollobasDegrees}.\\

2. Let $X$ be the number of vertices $i\in[n]$ such that 
\begin{equation}
    \left|\xi_i-\frac{n}{2}\right|>\frac{1}{2}\sqrt{n\ln n\left(1-\zeta_1\right)}-1.
\label{app:1}
\end{equation}
By the de Moivre--Laplace limit theorem, the probability of event~(\ref{app:1}) (for a fixed $i$) equals $1-\Phi(\sqrt{\ln n(1-\zeta_1)}+O(n^{-1/2}))+\Phi(-\sqrt{\ln n(1-\zeta_1)}+O(n^{-1/2}))$, where $\Phi(x)=\int_{-\infty}^x\frac{1}{\sqrt{2\pi}}e^{-t^2/2}dt$. Since $1-\Phi(x)\sim\frac{1}{\sqrt{2\pi}x}e^{-x^2/2}$  as $x\to\infty$ (see, e.g.,~\cite{Feller}) and due to the symmetry of $\Phi$, we get that
\begin{align}
\probb{\left|\xi_i-\frac{n}{2}\right|>\frac{1}{2}\sqrt{n\ln n\left(1-\zeta_1\right)}+O(1)}
&=(1+o(1))\sqrt{\frac{2}{\pi\ln n}}e^{-\frac{\ln n(1-\zeta_1)}{2}}\notag \\
&=
(1+o(1))\sqrt{\frac{2}{\pi n\ln n}}e^{\frac{\ln(2\pi)}{2}+\frac{3\ln\ln n}{4}}\notag \\
&=(2+o(1))\frac{(\ln n)^{1/4}}{\sqrt{n}}.
\label{app:2}
\end{align}
Moreover, for fixed distinct $i,j\in[n]$, the probability that (\ref{app:1}) holds for $i$ and $j$ simultaneously is
\begin{multline*}
     p_{i,j}=\left(\probb{\left|\mathrm{Bin}(n-2,1/2)-\frac{n}{2}\right|>\frac{1}{2}\sqrt{n\ln n(1-\zeta_1)}-2}\right)^2\prob{i\sim j}\\
     +\left(\probb{\left|\mathrm{Bin}(n-2,1/2)-\frac{n}{2}\right|>\frac{1}{2}\sqrt{n\ln n(1-\zeta_1)}-1}\right)^2\prob{i\nsim j}\\=(4+o(1))\frac{\sqrt{\ln n}}{n}
\end{multline*}
due to (\ref{app:2}) and since $\prob{i\sim j}+\prob{i\nsim j}=1$. Therefore, ${\sf E}X=(2+o(1))\sqrt{n}(\ln n)^{1/4}$ and 
$$
\mathrm{Var}X={\sf E}X+{\sf E}X(X-1)-({\sf E}X)^2={\sf E}X+n(n-1)p_{12}-({\sf E}X)^2=o(({\sf E}X)^2).
$$
By Chebyshev's inequality,
$$
 \prob{X>4\sqrt{n}(\ln n)^{1/4}}\leq(1+o(1))\frac{\mathrm{Var}X}{4n\sqrt{\ln n}}\to 0.
$$

3. Let $|x-1/2|\leq\frac{2\ln n}{n}$ and let $\tilde\xi_1=\tilde\xi_1(x),\ldots,\tilde\xi_n=\tilde\xi_n(x)\sim\mathrm{Bin}(n-1,x)$ be independent. For a non-negative integer $y$, let $\tilde N_y:=\tilde N_y(x)=\sum_{i=1}^n I(\tilde\xi_i=y)$ be the number of random variables equal to $y$.  By the (local) de Moivre--Laplace limit theorem, for an integer $y$ such that $|y-n/2|\leq\frac{1}{2}\sqrt{n\ln n(1-\zeta_2)}$,
$$
 \prob{\tilde \xi_1=y}=\prob{\mathrm{Bin}(n-1,x)=y}=\sqrt{\frac{2}{\pi n}}e^{-2\frac{(y-n/2)^2}{n}}\left(1+O\left(\sqrt{\frac{\ln^3 n}{n}}\right)\right).
$$
Therefore, 
\begin{equation}
{\sf E}\tilde N_y=\sqrt{\frac{2n}{\pi}}e^{-2\frac{(y-n/2)^2}{n}}\left(1+O\left(\sqrt{\frac{\ln^3 n}{n}}\right)\right)\geq 2(\ln n)^{5/4}(1+o(1)).
\label{app:3}
\end{equation}
By the Chernoff bound (Lemma~\ref{Chernoff}),
\begin{align*} 
\probb{|\tilde N_y-{\sf E}\tilde N_y|\geq\frac{1}{3\ln\ln n}{\sf E}\tilde N_y}
&\leq\exp\left[-(1+o(1))\frac{{\sf E}\tilde N_y}{18\ln^2\ln n}\right]\\
&\leq
 \exp\left[-(1+o(1))\frac{(\ln n)^{5/4}}{9\ln^2\ln n}\right].
\end{align*}

Therefore, letting $\mathcal{Y}_n$ to be the set of all integers $y$ such that $|y-n/2|\leq\frac{1}{2}\sqrt{n\ln n(1-\zeta_2)}$, we get
$$
 \probb{\exists y\in\mathcal{Y}_n\quad
 \tilde N_y<(\ln n)^{5/4}}
 \leq \sqrt{n\ln n} \exp\left[-(1+o(1))\frac{(\ln n)^{5/4}}{9\ln^2\ln n}\right]\to 0.
$$
Moreover, due to (\ref{app:3}), for every $y$ in the range, ${\sf E}\tilde N_{y+1}={\sf E}\tilde N_y(1+O(\ln^{3/2} n/\sqrt{n}))$. Then,
\begin{align}
 \probb{\exists y\in\mathcal{Y}_n\quad
 |\tilde N_y-\tilde N_{y+1}|\geq\frac{\tilde N_y}{\ln\ln n}}
 &\leq \probb{\exists y\in\mathcal{Y}_n\quad
 |\tilde N_y-{\sf E}\tilde N_{y}|\geq\frac{{\sf E}\tilde N_y}{3\ln\ln n}}\notag \\
 &\leq \sqrt{n\ln n} \exp\left[-(1+o(1))\frac{(\ln n)^{5/4}}{9\ln^2\ln n}\right]\to 0.
\label{app:4}
\end{align}
It remains to transfer the result from independent random variables to $\xi_1,\ldots,\xi_n$ using Theorem~\ref{th:MW} in the same way as in Section~\ref{sec:th1_proof_finish_probQ_small}. Let us do that.

Consider a normal random variable $\mathbf{p}$ with mean $1/2$ and variance $\frac{1}{4n(n-1)}$, truncated to $(0,1)$. Let $f$ be the density of $\mathbf{p}$. Then, for $n$ large enough, \refeq{p_is_near_1/2} holds. Therefore, due to (\ref{sum_is_even}) and Theorem~\ref{th:MW}, we get
\begin{multline*} 
\probb{\exists y\in\mathcal{Y}_n\quad
 N_y<(\ln n)^{5/4}}\\
  =\cprobb{\exists y\in\mathcal{Y}_n\quad
 \tilde N_y(\mathbf{p})<(\ln n)^{5/4} }{ \sum_{i=1}^{n}\tilde\xi_i(x)\text{ is even} }+o(1)\\
  =\int_0^1\cprobb{\exists y\in\mathcal{Y}_n\quad
 \tilde N_y(x)<(\ln n)^{5/4} }{ \sum_{i=1}^{n}\tilde\xi_i(x)\text{ is even}}f(x)dx+o(1)\\
 \leq 2\max_{|x-1/2|\leq 2\ln n/n}\probb{\exists y\in\mathcal{Y}_n\quad
 \tilde N_y(x)<(\ln n)^{5/4}}+o(1)=o(1).
\end{multline*} 

In the same way, \refeq{p_is_near_1/2}, (\ref{sum_is_even}), (\ref{app:4}) and Theorem~\ref{th:MW} imply that  
$$
\probb{\exists y\in\mathcal{Y}_n\quad
 |N_y-N_{y+1}|\geq\frac{N_y}{\ln\ln n}}\to 0
$$ 
as needed.\\

In the same way, by the (local) de Moivre--Laplace limit theorem, for an integer $y$ such that $\frac{1}{2}\sqrt{n\ln n(1-\zeta_2)}-2\leq|y-n/2|\leq\frac{1}{2}\sqrt{n\ln n(1-\zeta_1)}$+2,
\begin{equation}
2(\ln n)^{3/4}(1+o(1))\leq {\sf E}\tilde N_y\leq 2(\ln n)^{5/4}(1+o(1)).
\label{app:3.1}
\end{equation}
By the Chernoff bound (Lemma~\ref{Chernoff}),
$$
\probb{\tilde N_y>\ln^2 n}\leq\exp\left[-(1+o(1))\frac{3\ln^2 n}{2}\right].
$$
Therefore, by the union bound whp there is no integer $y$ such that $\frac{1}{2}\sqrt{n\ln n(1-\zeta_2)}-2\leq|y-n/2|\leq\frac{1}{2}\sqrt{n\ln n(1-\zeta_1)}+2$ and $\tilde N_y>\ln^2 n$. In the same way as above, \refeq{p_is_near_1/2}, (\ref{sum_is_even}), (\ref{app:4}) and Theorem~\ref{th:MW} imply that the same is true for $N_y$.\\

4. Let $\mathcal{Z}_n$ be the set of all integers $y$ such that 
$$
 \frac{1}{2}\sqrt{n\ln n(1-\zeta_2)}\leq |y-n/2|\leq\frac{1}{2}\sqrt{n\ln n(1-\zeta_1)}.
$$
Note that, for every $y\in\mathcal{Z}_n$, by the (local) de Moivre--Laplace limit theorem,
$$
{\sf E}N_y=\sqrt{\frac{2n}{\pi}}e^{-2\frac{(y-n/2)^2}{n}}\left(1+O\left(\sqrt{\frac{\ln^3 n}{n}}\right)\right)\geq (2+o(1))(\ln n)^{3/4}
$$ 
which can be proven exactly in the same way as (\ref{app:3}). By Lemma~\ref{lem:degrees_range}.3, whp, for every $y\in\mathcal{Z}_n$, $N_y\leq (\ln n)^2$. Therefore, it is sufficient to prove that whp the number of $y\in\mathcal{Z}_n$ such that $|N_y-{\sf E}N_y|>\frac{{\sf E}N_y}{3\ln\ln n}$ is at most $\frac{\sqrt{n}}{5\ln^2 n}$. Let us denote the number of such `bad' $y$'s by $W$. Borrowing notations from the proof of Lemma~\ref{lem:degrees_range}.3, we get by the Chernoff bound
$$
\probb{|\tilde N_y-{\sf E}\tilde N_y|\geq\frac{1}{3\ln\ln n}{\sf E}\tilde N_y}
\leq
 \exp\left[-(1+o(1))\frac{(\ln n)^{3/4}}{9\ln^2\ln n}\right].
$$
Let $\tilde W$ be the number of $y\in\mathcal{Z}_n$ such that $|\tilde N_y-{\sf E}\tilde N_y|\geq\frac{1}{3\ln\ln n}{\sf E}\tilde N_y$. Then, for $n$ large enough, 
$$
{\sf E}\tilde W\leq\frac{|\mathcal{Z}_n|}{\ln^4 n}=o\left(\frac{\sqrt{n}}{\ln^4 n}\right).
$$
By Markov's inequality, $\probb{\tilde W>\frac{\sqrt{n}}{5\ln^2 n}}=o\left(\frac{1}{\ln^2n}\right)$.  Relations \refeq{p_is_near_1/2}, (\ref{sum_is_even}) and Theorem~\ref{th:MW} imply that the same is true for $W$.\\

5. In the same way as in the proof of Lemma~\ref{lem:degrees_range}.3, the de Moivre--Laplace limit theorem, the Chernoff bound, relations \refeq{p_is_near_1/2}, (\ref{sum_is_even}) and Theorem~\ref{th:MW} imply that, whp, for every integer $y$ such that $|y-n/2|\leq\frac{1}{2}\sqrt{n(\ln n-2\sqrt{\ln n})}+1$, 
$$
N_y\geq\left(\sqrt{\frac{2}{\pi}}+o(1)\right)e^{\sqrt{\ln n}}.
$$
It remains to note that whp, in $\rg$ and its complement (which has the same distribution as $\rg$), there are no 
\begin{itemize}
\item cliques of size at least $2\ln n$, 
\item bipartite (not necessarily induced) subgraphs with both parts of size at least $3\ln n$.
\end{itemize}
Both statements follow from the union bound. The proof of the first one appears e.g. in~\cite[Chapter XI]{Bollobas-b}. To show that the latter is true, let us bound from above the expected number of bipartite graphs with both parts of size $\lfloor 3\ln n\rfloor$ by
$$
 n^{6\ln n}\left(\frac{1}{2}\right)^{(9+o(1))\ln^2 n}=\exp\left[(6-9\ln 2+o(1))\ln^2 n\right]\to 0
$$
as needed.

\end{document}